\documentclass[conference]{IEEEtran}

\let\labelindent\relax
\usepackage{enumitem}

\usepackage[english]{babel}
\usepackage[utf8]{inputenc}
\usepackage{amsthm} 
\usepackage{amsmath}
\usepackage{graphicx}
\usepackage{boxedminipage}
\usepackage{epsfig}
\usepackage{psfrag}
\usepackage{url}
\usepackage{comment}
\usepackage{afterpage}
\usepackage{pstricks}
\usepackage{graphics}

\usepackage{mdwtab}
\usepackage{wrapfig}
\usepackage{graphics}

\usepackage{algorithm}  
\usepackage{algpseudocode}

\usepackage{amssymb} 

\usepackage{makecell}

\usepackage{hhline}

\usepackage{tabularx}

\usepackage{epsfig}
\usepackage{fixfoot}
\usepackage{listings}
\usepackage{psfrag,wrapfig}
\usepackage{xspace,soul}
\usepackage[colorinlistoftodos]{todonotes}

\usepackage{subcaption}

\newcommand{\eg}{{\it e.g.,}\xspace}

\newcommand{\etal}{{\it et~al.}}
\newcommand{\ie}{{\it i.e.,}\xspace}\soulregister{\ie}{7}
\newcommand{\etc}{{\it etc.}}
\newcommand{\ci}{{\it (i) }}
\newcommand{\cii}{{\it (ii) }}

\newcommand{\ca}{{\it (a) }}
\newcommand{\cb}{{\it (b) }}

\newcommand{\wrt}{{\it w.r.t.}\xspace}

\usepackage{amsthm}
\theoremstyle{definition}
\newtheorem{definition}{Definition}
\newtheorem{example}{Example}
\newtheorem{proposition}{Proposition}
\newtheorem{theorem}{Theorem}

\newcommand{\mmod}{\;\mathrm{mod}\;}
\renewcommand{\qedsymbol}{\hfill$\blacksquare$}

\newcommand{\ATTCK}{{\sffamily DyPS}\xspace}
\soulregister{\ATTCK}{7} 
\newcommand{\ATTCKNF}{{DyPS}\xspace}
\soulregister{\ATTCKNF}{7} 

\begin{document}

\title{On Scheduler Side-Channels \\ in Dynamic-Priority Real-Time Systems}

\author{
\IEEEauthorblockN{Chien-Ying Chen\IEEEauthorrefmark{1}, Sibin Mohan\IEEEauthorrefmark{1}, Rodolfo Pellizzoni\IEEEauthorrefmark{2} and Rakesh B. Bobba\IEEEauthorrefmark{3}} 
\IEEEauthorblockA{\IEEEauthorrefmark{1}Deptartment of Computer Science, University of Illinois at Urbana-Champaign, Urbana, IL, USA}
\IEEEauthorblockA{\IEEEauthorrefmark{2}Deptartment of Electrical and Computer Engineering, University of Waterloo, Ontario, Canada}
\IEEEauthorblockA{\IEEEauthorrefmark{3}School of Electrical Engineering and Computer Science, Oregon State University, Corvallis, OR, USA}
Email: \{\IEEEauthorrefmark{1}cchen140,
\IEEEauthorrefmark{1}sibin\}@illinois.edu,
\IEEEauthorrefmark{2}rodolfo.pellizzoni@uwaterloo.ca,
\IEEEauthorrefmark{3}rakesh.bobba@oregonstate.edu
}

\maketitle

\thispagestyle{plain}
\pagestyle{plain}

\begin{abstract}
While the existence of scheduler side-channels has been demonstrated recently for fixed-priority real-time systems (RTS), there have been no similar explorations for dynamic-priority systems.
The dynamic nature of such scheduling algorithms, e.g., EDF, poses a significant challenge in this regard.
In this paper we demonstrate that \textit{side-channels exist in dynamic priority real-time systems}.
Using this side-channel, our proposed \ATTCKNF algorithm is able to effectively infer, with high precision, critical task information from the vantage point of an unprivileged (user space) task. 
Apart from demonstrating the effectiveness of \ATTCKNF, we also explore the various factors that impact such attack algorithms using a large number of synthetic task sets. 
%
We also compare against the state-of-the-art and demonstrate that our proposed \ATTCKNF algorithms outperform the ScheduLeak algorithms in attacking the EDF RTS.

\end{abstract}

\section{Introduction}
\label{sec::intro}

Due to the increased deployment of safety-critical systems with timing criticality (\eg autonomous cars, delivery drones, industrial robots, implantable medical devices, power grid components), security for such systems becomes crucial.
Until recently, security has been an afterthought in the design of real-time systems (RTS).
However, the ever-increasing demand for using commodity-off-the-shelf (COTS) components and the demonstration of a variety of attacks against such systems in the field~\cite{chen2011stuxnet, case2016analysis, JeepHacking101, byungho2014attack, yoon2017virtualdrone, auto:koscher2010,DroneHack:Shepard2012,embeddedsecurity:teso2013} necessitates the need for a better understanding and classification of security threats aimed at RTS.

Side-channels that leak critical information about task behavior via system schedules in RTS has recently gained attention, including methods to protect against them~\cite{2016:taskshuffler, kruger2018vulnerability, nasri2019pitfalls, chen2019novel, liu2019leaking}.
In particular, the \textit{ScheduLeak} algorithms~\cite{chen2019novel} demonstrate that scheduler side-channels can be exploited by an unprivileged task (``observer task'') to leak important information such as task arrival times in fixed-priority RTS (FP RTS).
This information was then used to predict future arrival times of critical tasks (``victim tasks'').
While the leakage of such information seems to be subtle, knowledge about future arrival times, especially for critical tasks, can help increase the effectiveness of other attacks by filtering out noisy data and extracting valuable information about the victim system.
In fact, these types of attacks fall into the broader category of \textit{reconnaissance attacks}.
It has been shown that they can be the stepping stone for more sophisticated attacks~\cite{tankard2011advanced,virvilis2013big}.
For instance, the ScheduLeak paper demonstrated how cache timing attacks or even the ability to take control of autonomous drones become much simpler once the victim task's future behavior is made available to adversaries.
One main drawback of ScheduLeak is that it has only been demonstrated for FP RTS.
This significantly limits the types of systems where such attacks can be launched.
Directly applying ScheduLeak to dynamic-priority RTS also does not work well --- as evidenced by the precision of inference (for EDF; the grey bars) in Figure \ref{fig:duration_sleak_in_edf}.
Hence, we need to \textit{develop algorithms that are targeted towards dynamic-priority real-time systems}.

\begin{figure}[htb]
    \vspace{-0.5\baselineskip}
    \centering
    \includegraphics[width=0.75\columnwidth]{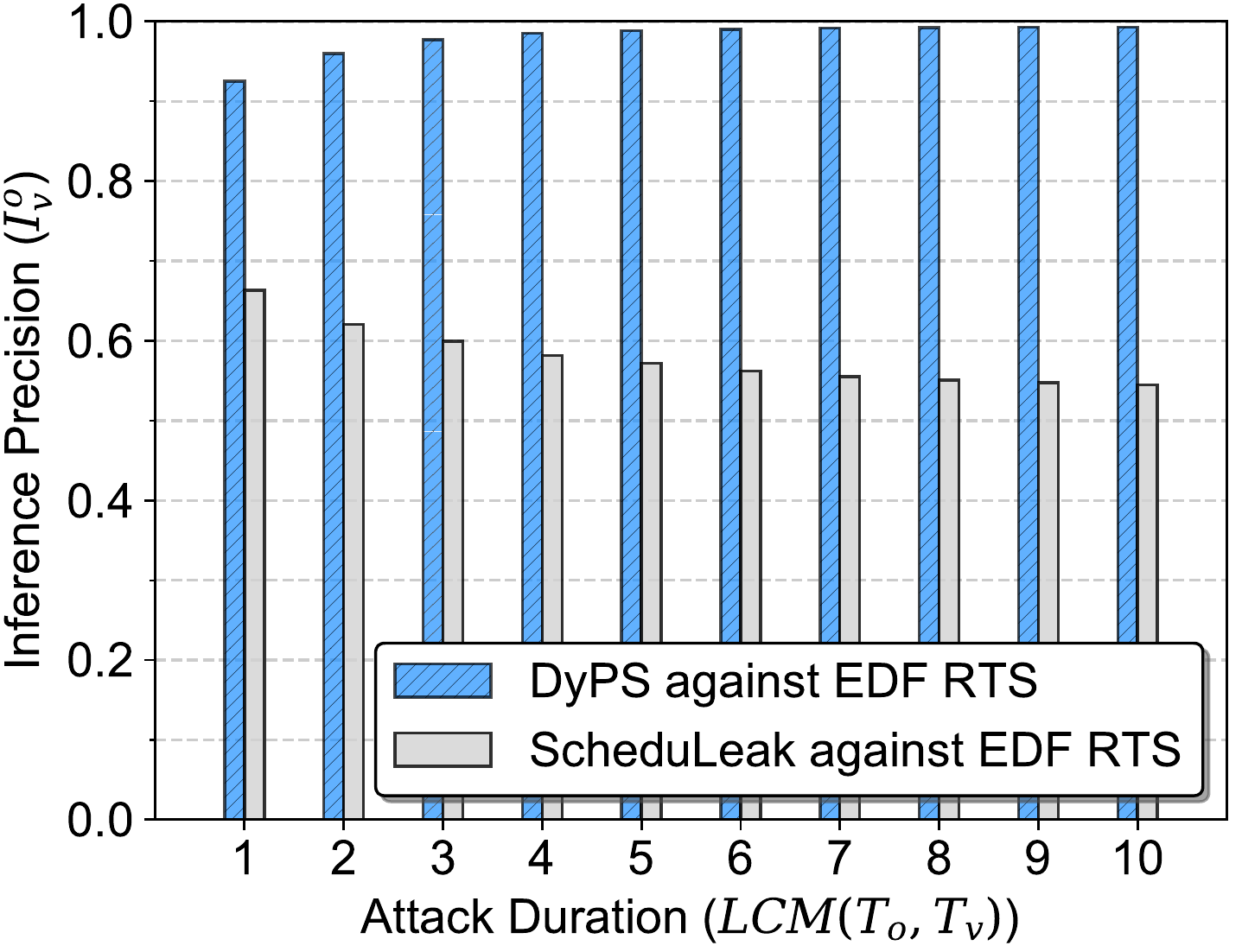}
    \caption{The results of employing the state-of-the-art (\text{ScheduLeak}) and the proposed \ATTCKNF algorithms in the scheduler side-channel attack against the EDF RTS. The end (mean) inference precision of ScheduLeak is $0.54$ that is only slightly better than a naive attack with random guesses.}
    \label{fig:duration_sleak_in_edf}
     \vspace{-0.5\baselineskip}
\end{figure}
One important challenge to leaking information via scheduler side-channels in dynamic priority RTS, \eg the earliest-deadline first (EDF) algorithm, is that (relative) task priorities are not constant and vary at run-time.
In the ScheduLeak attack~\cite{chen2019novel} that targets FP RTS, the observer task has a priority lower than the victim task at all times.
This determinism ensures that the execution of the observer task is always preempted or delayed by the victim task when both are ready to run and this is vital for inferring the arrival times of the victim.
In contrast, in an EDF RTS, task priorities are determined dynamically based on each job's absolute/relative deadline at each scheduling point.
That is, no task will always have a higher priority relative to another task in the system. 
Hence, the ScheduLeak assumption about a persistent (relative) priority relationship between any two tasks in the system becomes invalid in EDF RTS.
Consequently, the ScheduLeak algorithms designed for FP RTS will fail (or have much lower success rates) while dealing with EDF RTS.

We present the \underline{Dy}namic-\underline{P}riority \underline{S}cheduLeak (\ATTCKNF) algorithms that
\ca demonstrate the existence of scheduler side-channels in the Earliest Deadline First (EDF) scheduling algorithm and
\cb make use of such information to extract critical task information.
\ATTCKNF builds on top of ScheduLeak but introduces additional steps and analyses to overcome the challenges \wrt dynamic priority schedulers.
We then present our findings on how to craft effective attacks --- that builds upon the pairwise relationships between tasks in the system.
In addition, we also explore conditions that could limit the scope of this attack.
Finally, our evaluation demonstrates the effectiveness of \ATTCKNF and also compares it to the state-of-the-art (ScheduLeak) --- we are able to match the performance of ScheduLeak with a difference in precision of less than $0.1\%$.

To summarize, this paper makes the following contributions:
\begin{enumerate}
  \item \ATTCKNF \ --- a set of attack algorithms that overcome the uncertainty in the EDF scheduling to accurately extract critical task information from the scheduler side-channels [Section~\ref{sec:scheduler_side_channels}];

  \item Analyses and metrics to understand the factors that can influence an attacker's ability to carry out a successful \ATTCKNF attack [Section~\ref{sec:analysis}];

  \item Comparison with the state-of-the-art scheduler side-channel attack algorithms  [Section~\ref{sec:eval}].
\end{enumerate}

\section{System and Adversary Model}

\subsection{System Model}
\label{sec:system_model}

In this paper, a discrete time model \cite{isovic2001handling} is considered.
We assume that there exists a system timer (\eg a $\mathtt{CLOCK\_MONOTONIC}$-based timer in Linux, a global tick counter in FreeRTOS) 
that produces time ticks that the EDF scheduler can use and that the tick count is an integer. 
We further assume that a unit of time is equal to a time tick and all system parameters are multiples of this time tick.
We denote an interval starting from time point $a$
and ending at time point $b$ by $[a,b)$ or $[a,b-1]$ (and hence its length equals $b-a$.)

We consider a uni-processor, single-core, preemptive, dynamic-priority RTS running the EDF scheduling algorithm. The system consists of $n$ real-time tasks $\Gamma=\{\tau_1,\tau_2...,\tau_n\}$, each of which can be either a periodic or a sporadic task\footnote{A task may also be an aperiodic task. However, in systems like real-time Linux (\ie Linux with the $\mathtt{PREEMPT\_RT}$ patch), aperiodic tasks only get to run in slack time (\ie when no real-time tasks are in the ready queue). As a result, aperiodic tasks do not influence how real-time tasks behave and thus are ignored in this paper.}. 
A task $\tau_i$ is modeled by $C_i$, $T_i$, $D_i$, $\phi_i$ where $C_i$ is the worst-case execution time (WCET), $T_i$ is the period (or the minimum inter-arrival time for a sporadic task), $D_i$ is the relative deadline and $\phi_i$ is the task phase\footnote{The task phase is defined as the offset from the zero time point to any of the task's arrival time points projected on the period on the zero time point. Thus, $\phi_i<T_i$. It should not be confused with the arrival time point of the task's first job.}.
We assume that every task has a distinct period (or the minimum inter-arrival time) and that $D_i=T_i$~\cite{LiuLayland1973}.
We denote the $k$-th job of the task $\tau_i$ by $\tau_i^k$ and it is modeled by $c^k_i$, $a^k_i$, $s^k_i$, $d^k_i$ where $c^k_i$ denotes the execution time ($c^k_i \leq C_i$), $a^k_i$ is the absolute arrival time, $s^k_i$ is the start time and $d^k_i$ is the absolute deadline.
For simplicity, we use $c_i$, $a_i$, $s_i$, $d_i$ when referring to an arbitrary job of $\tau_i$ if the job ordering is unimportant. Furthermore, we use ``task'' and ``job'' interchangeably. 
In this paper, we only consider the task set that is \textit{schedulable} by the EDF scheduling algorithm. Therefore, $\sum_{\tau_i \in \Gamma} \dfrac{C_i}{T_i} \leq 1$. 
We assume that the task release jitter is negligible, and thus $a^{k+1}_i-a^{k}_i=T_i$ for a periodic task and $a^{k+1}_i-a^{k}_i \geq T_i$ for a sporadic task. 

In the EDF scheduling algorithm, a job with smaller absolute deadline gets to run first and is considered to have higher priority among other ready jobs that have greater absolute deadlines.
In the case that multiple jobs in the ready queue have the same absolute deadline, they are considered to have the same priority and the EDF scheduler \textit{randomly} selects one of the jobs to run.
%
We define a task's ``execution interval'' to be an interval during which the task runs continuously.
\subsection{Adversary Model}
Similar to the adversary model introduced in existing work on scheduler side-channels in RTS~\cite{chen2019novel, 2016:taskshuffler, liu2019leaking}, we assume that the attacker is interested in learning the task phase (and then inferring the future arrival time points) of a critical, \textit{periodic} task (the victim task, denoted by $\tau_v$) in the system. 
This is considered to be part of a reconnaissance phase that can be a part of a larger attack. Such an attack will benefit from the inferences of the task's future arrival time points. 
The attacker launches
the proposed \ATTCKNF attack algorithms
that exploit the scheduler side-channels using an unprivileged, periodic task (the observer task, denoted by $\tau_o$) running on the same victim system. 
Once the inference of the victim task's phase is obtained, it is up to the attacker to decide if further attacks should be launched using the same observer task or via other attack surfaces.
The ultimate goal of the attacks varies with the adversaries.
It has been shown that the inferred future arrival time points can be used to help increase the success rate of an attack (\eg overriding the control of a rover system) or filtering out noisy data while extracting valuable information (\eg monitoring the task's execution behavior by using a cache-timing side-channel attack)~\cite{chen2019novel}.

The \ATTCKNF attack requires only the observer task to ensure the success for the attack algorithms introduced in this paper.
We assume that the observer task has access to a system timer that has a resolution that is coarser than or equal to a time tick.
The observer task uses the timestamps read from such a system timer to reconstruct its own execution intervals and infer the victim task's phase.
However, this method only works when the victim task has a priority higher than the observer task~\cite[Theorem 1]{chen2019novel}, which is not always true under the EDF scheduling algorithm.
More specifically, in the case of EDF RTS, we are more interested in the priority relationships between tasks at \textit{the instant when the victim task arrives}. 
To better clarify the relation between the observer task and the victim task, we define the term ``observability'' as follows:

\begin{definition}
\label{def:observability}
(Observability) A victim task's arrival at $a_v$ is said to be observable by the observer task if the observer task has a priority lower than the victim task at $a_v$. 
\qedsymbol
\end{definition}

%
In an FP RTS, it is trivial to see that every arrival of the victim task is observable by the observer task if the victim task has a fixed, higher priority than the observer task.
In an EDF RTS, it depends on both the tasks' periods and the absolute deadlines at run-time. We first determine the observability of a given observer and victim task pair by using the following theorem:

\begin{theorem}
\label{th:ToObserveTv}
%
For the EDF scheduling algorithm, given an observer task $\tau_o$ and a victim task $\tau_v$ whose periods are $T_o$ and $T_v$ respectively and $T_o \neq T_v$, the victim task's arrivals may be observable by the observer task only if $T_o > T_v$.
\end{theorem}
\begin{proof}
Definition~\ref{def:observability} for EDF means that the observer task has a larger deadline when the victim task arrives.
That is, $a_o \leq a_v$ and $d_o>d_v$ (or $a_o+T_o>a_v+T_v$) for some jobs of $\tau_o$ and $\tau_v$.
Rewriting the above deadline inequality as $T_o-T_v>a_v-a_o$ indicates that $T_o$ must be greater than $T_v$ since $a_v-a_o \geq 0$.
%
Now let's consider the case when $T_o<T_v$. It can be seen that when both tasks arrive at the same time (which is the case when the observer task has the largest possible deadline relative to the victim task), the observer task still has a deadline smaller than the victim task (\ie $a_o=a_v$ and thus $a_o+T_o<a_v+T_v \Rightarrow d_o<d_v$). Therefore, no victim task's arrival can be observed by the observer task if $T_o<T_v$.
\end{proof}

Based on Theorem~\ref{th:ToObserveTv}, we make an assumption that the observer task must have a period larger than the victim task (\ie $T_o>T_v$) in this paper.

\section{The Scheduler Side-Channels in EDF}
\label{sec:scheduler_side_channels}

The scheduler side-channels in the FP RTS enable an unprivileged, low-priority task to learn precise timing information of a periodic, critical task.
Similar scheduler side-channels exist in the EDF RTS due to the fact that both types of RTS are preemption-based systems. 
However, because of the dynamic nature, there are additional conditions and restrictions to be considered for making the attack succeed under the EDF scheduling algorithm.
In this section, we present these constraints along with the details of the proposed \ATTCKNF algorithms.

\subsection{Challenges and Overview}
\label{sec:challenges_and_overview}
In this paper, the attacker's goal is to infer the victim task's phase and then predict its future arrival time points. 
This is achieved by allowing the observer task to reconstruct and analyze \textit{its own execution intervals}.
When the victim task arrives with a priority higher than the job of the observer task that has been scheduled, the execution of the latter is either delayed or preempted. As a result, the victim task's execution (including the arrival instant) is enclosed in the observer task's execution intervals.
By reconstructing execution intervals for a sufficiently long duration (see Section~\ref{sec:eval_infer_phiv} for the evaluation of the attack duration), it is possible to infer the victim task's phase. 
Yet, even if we assume that $T_o > T_v$ (Theorem~\ref{th:ToObserveTv}) for a given observer and victim task pair, it is not guaranteed that every arrival of the victim task is observable by the observer task due to the dynamic priority in EDF.

\begin{example}
\label{ex:counterToObserveTv}
Consider an observer task $\tau_o$ ($T_o=10$, $C_o=4$) and a victim task $\tau_v$ ($T_v=8$, $C_v=2$). 
Without making any assumption on the task phases,
Figure~\ref{fig:ex_observable} demonstrates how dynamic priorities impact the observability in EDF schedules.
In Figure~\ref{fig:ex_observable}(a), both $\tau_o$ and $\tau_v$ arrive at the same time point (\ie $a_o=a_v$). Since $\tau_v$ has a higher priority (because $d_v<d_o$), the execution of $\tau_v$ delays $\tau_o$ that is supposed to start at $a_o$. As a result, the arrival at $a_v$ can be observed by $\tau_o$.
In Figure~\ref{fig:ex_observable}(b), $\tau_v$ arrives at a point when $d_o<d_v$. As $\tau_o$ is currently executing and has a higher priority at the arrival point $a_v$, the execution of $\tau_v$ is delayed by the execution of $\tau_o$. Consequently, the observer task fails to observe the victim task's arrival at $a_v$.
\qedsymbol

\end{example}

\begin{figure}[t]
    \centering
    \begin{subfigure}[t]{0.99\columnwidth}
        \centering
        \includegraphics[width=0.99\columnwidth]{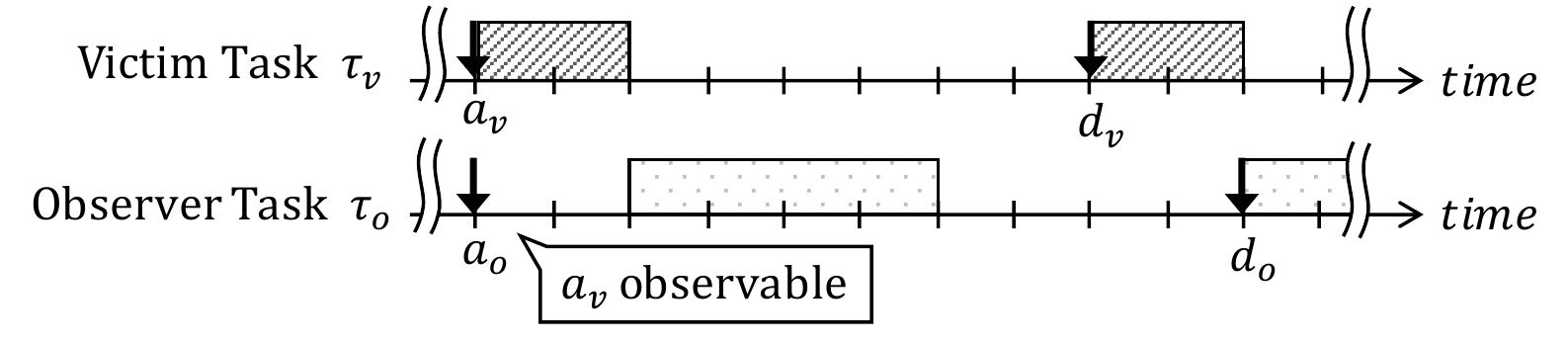}
        \caption{$d_v<d_o$ and thus $a_v$ is observable by $\tau_o$}
    \end{subfigure}%
    \vspace{1\baselineskip}
    \begin{subfigure}[t]{0.99\columnwidth}
        \centering
        \includegraphics[width=0.99\columnwidth]{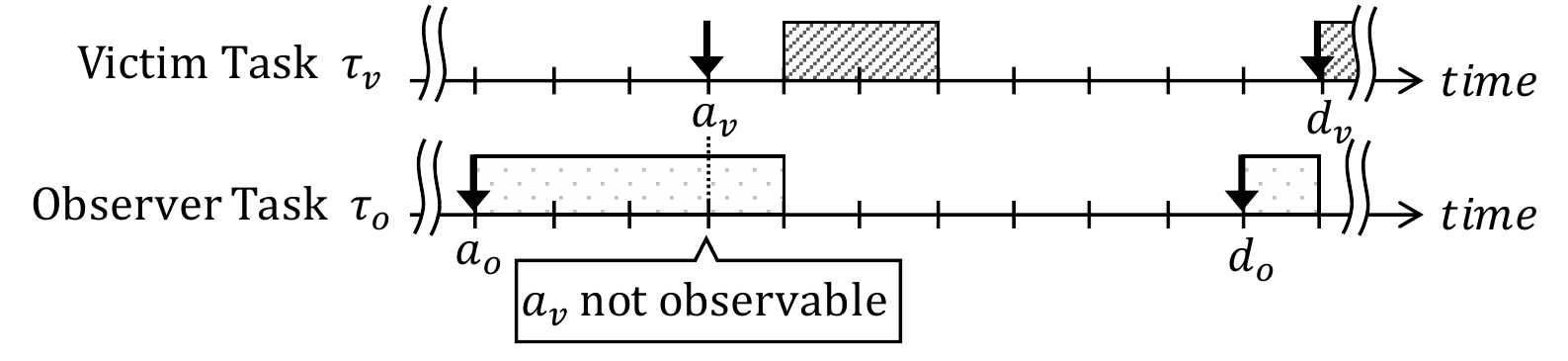}
        \caption{$d_v>d_o$ and thus $a_v$ is not observable by $\tau_o$}
    \end{subfigure}
    \caption{Examples of the two conditions (the victim task's arrival, $a_v$, being observable and not observable by the observer task under the EDF scheduling) elaborated in Example~\ref{ex:counterToObserveTv}.}
    \label{fig:ex_observable}
     \vspace{-1\baselineskip}
\end{figure}

Apart from the fact that not all the observer task's execution intervals encapsulate the victim task's arrivals, the above example also hints at the fact that a part of an execution interval may still be admissible even when such an execution interval is considered invalid \wrt the observation of the victim task's arrivals.
To illustrate, let's consider the observer task's execution interval starting at $a_o$ in Figure~\ref{fig:ex_observable}(b). It is in fact safe to leverage the first half of the execution interval to infer there is no arrival. This is because if the victim task arrived within this interval it would have higher priority (earlier deadline) and would preempt the observer task. 
But considering the second half of the execution interval to make the same inference would cause a false negative observation (\ie no arrival is inferred while there is one) and mislead the attack results.
To identify the part of the execution interval that is valid for observing the victim task's arrivals, we present the following theorem:

\begin{theorem}
\label{th:valid_observation_region}
For a given job of the observer task arriving at $a_o$, only the execution interval(s) within $[a_o, a_o+T_o-T_v)$ is valid for observing the arrivals of the victim task.
\end{theorem}
\begin{proof}
The observer task may observe the victim task's arrivals if $a_o \leq a_v$ and $d_o>d_v$.
However, it is unknown to the attacker when a job of the victim task would arrive, so the attacker must assume that the victim task may arrive at any time point and exclude the part of the execution intervals that may contain false information. 
For the jobs of $\tau_o$ whose deadlines satisfy $d_o>d_v$, all the execution intervals in such jobs provide valid observations of the victim task's arrivals. That is, for a given $a_o$ of a job whose deadline meets $d_o>d_v$, all the execution intervals within the range $[a_o, a_o+T_o)$ are valid.
In contrast, when $d_o \leq d_v$, the execution of $\tau_o$ may interfere\footnote{By ``$\tau_i$ interferes with an arrival of $\tau_j$'' we mean that $\tau_i$ has a priority higher than $\tau_j$ at the arrival point of $\tau_j$ and hence the start of $\tau_j$ will be delayed by the execution of $\tau_i$.} with the victim task's arrivals. From $\tau_o$'s point of view, the earliest victim task's arrival that may be interfered by $\tau_o$ occurs when $d_o=d_v$ and that arrival time point can be represented by $a_v=d_v-T_v=d_o-T_v=a_o+T_o-T_v$. 
Hence, it is possible for the observer task's execution to interfere with any arrivals of the victim task if the execution spans across the time point $a_o+T_o-T_v$. 
Therefore, only the execution intervals within $[a_o, a_o+T_o-T_v)$ is valid for the observer task to observe the victim task's arrivals.
\end{proof}

Using Theorem~\ref{th:valid_observation_region}, it is possible for the observer task to reconstruct only the valid part of the execution intervals. However, the attacker may not directly use such a theorem as it requires $a_o$ (or more precisely, the task phase $\phi_o$) to be known. 
If the attacker is already present when the system starts, $\phi_o$ may be known to the attacker. However, in most attack cases where the attacker enters the victim system after the system starts, the attacker may not be able to easily learn the exact value of $\phi_o$ without further reconnaissance.
In such cases, the attacker must first obtain $\phi_o$ before employing Theorem~\ref{th:valid_observation_region} and proceeding to infer the victim task's phase. 

\begin{figure}[t]
    \centering
    \includegraphics[width=0.75\columnwidth]{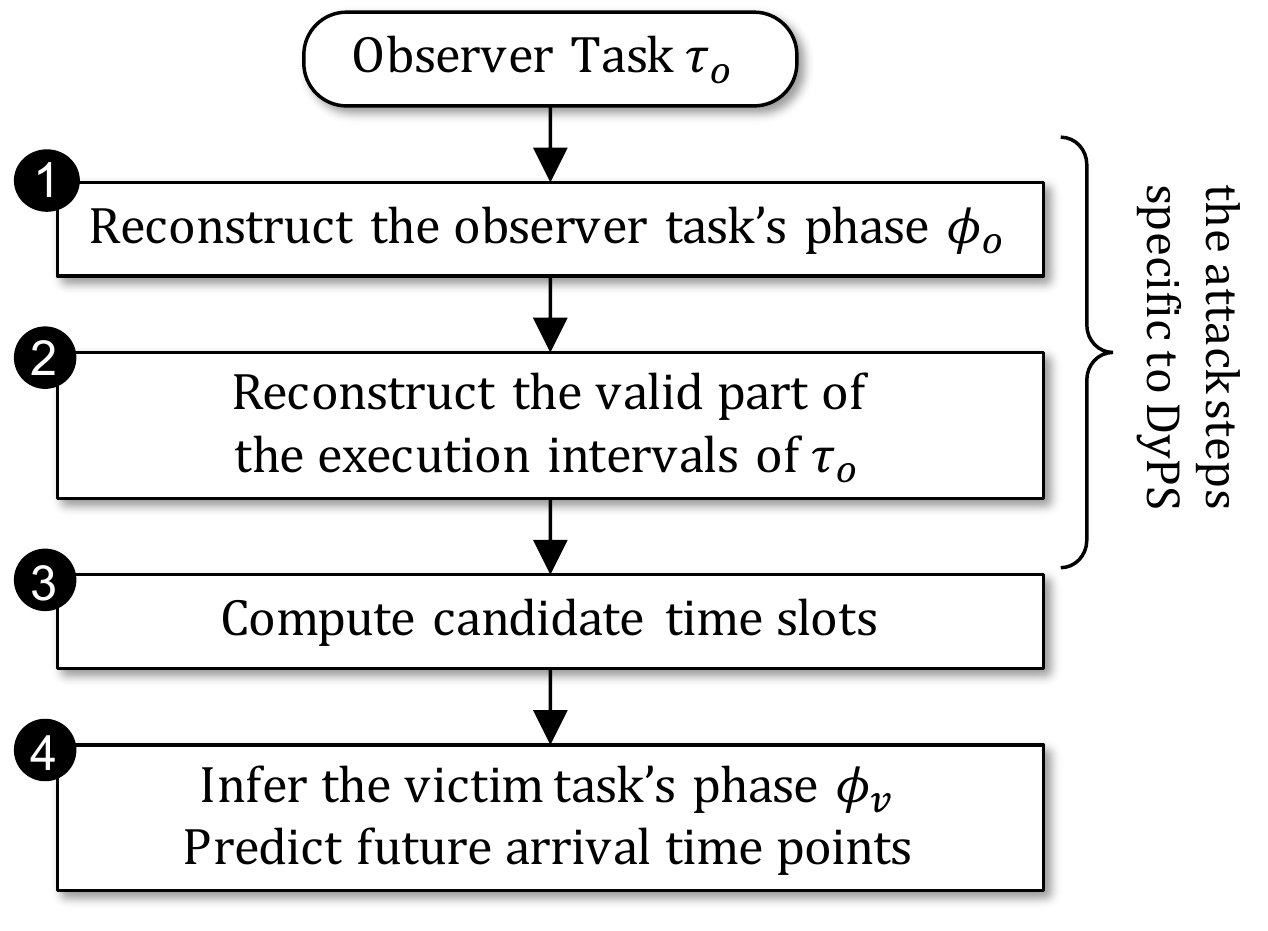}
    \caption{An overview of the attack steps in the \ATTCKNF algorithms. The first two steps are specific to dealing with the dynamic nature in EDF. Step 2 and 3 are identical to ScheduLeak.}
    \label{fig:dyps_steps}
     \vspace{-1.25\baselineskip}
\end{figure}

The rest of this section details every step of the proposed \ATTCKNF algorithms that account for the aforementioned challenges.
An overview of the attack steps is shown in Figure~\ref{fig:dyps_steps} and we present a brief description for the goal in each step. Note that the first two steps are unique to \ATTCKNF for tackling the dynamic nature in EDF, while the last two steps are identical to ScheduLeak~\cite{chen2019novel}, since we build on top of it.

\vspace{0.2\baselineskip}
\noindent
[\textbf{Step 1}] Reconstruct $\phi_o$: the first step is to reconstruct the observer task's phase in order to identify the range specified in Theorem~\ref{th:valid_observation_region}. [Section~\ref{sec:infer_phi_o}]

\vspace{0.2\baselineskip}
\noindent
[\textbf{Step 2}] Reconstruct execution intervals: with the reconstructed $\phi_o$, the observer task then is able to reconstruct only the valid part of the execution intervals for observing the victim task's arrivals. [Section~\ref{sec:reconstruct_execution_intervals}]

\vspace{0.2\baselineskip}
\noindent
[\textbf{Step 3}] Compute candidates: the observer task analyzes the reconstructed execution intervals and compute a list of time points as candidates for the final inference. [Section~\ref{sec:compute_candidates}]

\vspace{0.2\baselineskip}
\noindent
[\textbf{Step 4}] Infer $\phi_v$ and predict future $a_v$: a time point is selected from the candidate list as the inference of the victim task's phase. A future arrival time point of the victim task can then be predicted by using the inferred task phase. [Section~\ref{sec:infer_phiv}]


\subsection{Reconstructing The Observer Task's Phase}
\label{sec:infer_phi_o}
Reconstructing the observer task's phase, $\phi_o$, can be carried out by the observer task itself.
When a new job of the observer task is scheduled to run, $a_o$ is unknown since there might be higher priority tasks delaying the observer task's execution. However, what the observer task itself can learn is the job's start time $s_i$ (by reading the time stamp right as the job starts) which is bounded by $a_o \leq s_o \leq d_o-C_o$.
Let $\widetilde{\phi_o}$ be the reconstructed task phase of the observer task.
For a given job, we may compute $\widetilde{\phi_o}$ from the start time $s_o$ by
\begin{equation}
\label{eq:infer_phio}
    \widetilde{\phi_o} = s_o \mmod T_o
\end{equation}
where $T_o$ is known to the attacker. 
Intuitively, the closer $s_o$ is to $a_o$, the more accurate $\widetilde{\phi_o}$ will be.
By collecting and examining more start times, the attacker may further improve $\widetilde{\phi_o}$ by first determining the closest $s_o$ to $a_o$ and then computing $\widetilde{\phi_o}$ using Equation~\ref{eq:infer_phio}. 
When there exists one job whose $s_o$ is equal to $a_o$, the correct task phase can be reconstructed (\ie $\widetilde{\phi_o}=\phi_o$).
In Section~\ref{sec:analysis_phio} we discuss the factors that impact the reconstruction of the task phase and how likely it is for the attacker 
(observer task) to observe a situation where $s_o=a_o$ in a given task set.
We now formally describes the algorithms.

Consider the observer task launching the attack on its $k$-th job ($k$ is an arbitrary number that is unknown to the attacker) and collecting its own job start times for $m$ jobs. What the observer task captures is a set of start times of consecutive jobs $S^{observed}_o=\{s^k_o, s^{k+1}_o,...,s^{k+m-1}_o\}$.
Our goal is to find the start time that is closest to its arrival time.
We do this by comparing the start times using the following proposition:

\begin{proposition}
\label{prop:closer_s}
Given two start times $s^k_o$ and $s^{k+p}_o$ where $p \geq 1$ and thus $s^k_o < s^{k+p}_o$, we can determine the start time that is closer to its arrival time to be
\begin{equation}
    \begin{cases}
    s^k_o & \text{if } s^k_o < s^{k+p}_o-p \cdot T_o\\
    s^{k+p}_o & \text{otherwise.}
\end{cases}
\end{equation}
where $T_o$ and $p$ are known. 
\qedsymbol
\end{proposition}
In the above proposition, the given start time pair, $s^k_o$ and $s^{k+p}_o$, represents two jobs differing in $p$ periods.
Therefore, the start time of the $(k+p)$-th job can be shifted to the same period as the $k$-th job by $s^{k+p}_o-p \cdot T_o$ since the observer task arrives periodically (\ie $a^k_o=a^{k+p}_o-p \cdot T_o$).
As a result, the two start times in the same period ($s^k_o$ and the shifted $s^{k+p}_o-p \cdot T_o$) are comparable.
The smaller one is closer to its arrival time since $a_o \leq s_o$.

By employing Proposition~\ref{prop:closer_s}, we can determine the start time that is closest to its arrival time using $S^{observed}_o$. The inference of the task phase for the observer task is computed by
\begin{equation}
\label{eq:infer_phi_o_full}
    \widetilde{\phi_o}=\min(s_o^{k+p}-p \cdot T_o \mid 0 \leq p < m) \mmod T_o
\end{equation}
where $T_o$ is the period known to the attacker and $k$ can be unknown.
Then, given a start time $s_o$, the attacker can compute its projected arrival time $\widetilde{a_o}$ by
\begin{equation}
\label{eq:inferred_a_o}
    \widetilde{a_o}=s_o-\widetilde{j_o}
\end{equation}
where 
$\widetilde{j_o}=(s_o-\widetilde{\phi_o}) \mmod T_o$ 
represents the delay such a job may have experienced.

\begin{example}
\label{ex:infer_phi_o}
Consider the observer task $\tau_o$ in a task set of 4 periodic tasks (extended from Example~\ref{ex:counterToObserveTv}) as shown in the table below.

\begin{center}\footnotesize
\vspace{-0.5\baselineskip}
\begin{tabular}{|c||c|c|c|}
\hline 
 & $T_i$ & $C_i$ & $\phi_i$ \\ 
\hline \hline 
$\tau_1$ & 15 & 1 & 3 \\ \hline 
$\tau_o$ & 10 & 4 & 1 \\ \hline 
$\tau_v$ & 8 & 2 & 2 \\ \hline
$\tau_4$ & 6 & 1 & 4 \\ \hline 
\end{tabular} 
\end{center}

Let's assume the system begins at $t=0$
and the observer task starts collecting its start times for $10$ instances from $t=41$.
The collected start times are 
$\{41,53,61,71,81,92,101,111,121,133\}$. 
By using Equation~\ref{eq:infer_phi_o_full}, the observer task's phase $\widetilde{\phi_o}$ can be computed by 
$\min(41,43,41,41,41,42,41,41,41,43) \mmod 10=1$.
In this example, $\widetilde{\phi_o}=\phi_o=1$ and thus the correct observer task's phase is obtained.
\qedsymbol
\end{example}

\subsection{Reconstructing The Observer Task's Execution Intervals}
\label{sec:reconstruct_execution_intervals}
Based on Theorem~\ref{th:valid_observation_region} and Equation~\ref{eq:inferred_a_o}, we developed the following proposition to reconstruct the execution intervals in a period:

\begin{proposition}
\label{prop:reconstruct_execution_intervals}
Assume that, in a \ATTCKNF attack, the attacker has reconstructed the observer task's phase as $\widetilde{\phi_o}$. 
Then, for a given job of the observer task starting at $s_o$, the observer task reconstructs only the execution intervals within $[\widetilde{a_o}, \widetilde{a_o}+T_o-T_v)$ where $\widetilde{a_o}$ is calculated using Equation~\ref{eq:inferred_a_o}.
\qedsymbol
\end{proposition}

To implement Proposition~\ref{prop:reconstruct_execution_intervals}, we employ a known reconstruction algorithm~\cite[Algorithm 1]{chen2019novel} but make the following modifications:
\ci at the beginning of the execution of each period, we let the observer task compute $\widetilde{a_o}$ by using Equation~\ref{eq:inferred_a_o}; 
\cii the observer task stops reconstructing execution intervals in a period if the current time stamp exceeds $\widetilde{a_o}+T_o-T_v$.

The end result of this step is a set of reconstructed execution intervals, denoted by $E^{recon}_o=\{e^1_o,e^2_o,e^3_o,...\}$ where $e^r_o := [begin^r_o,end^r_o)$ is the $r$-th reconstructed execution interval starting at the time point $begin^r_o$ and ending at the time point $end^r_o$.
Note that the number of reconstructed execution intervals are dependent on the duration of the attack that is determined by the attacker. The impact of the attack duration is evaluated in Section~\ref{sec:eval_infer_phiv}.

\subsection{Computing The Candidates}
\label{sec:compute_candidates}

To compute the candidate time points for the phase of the victim task, the reconstructed execution intervals are organized on 
a timeline with length equal to the victim task's period $T_v$. 
To facilitate understanding, let us use the \textit{schedule ladder diagram}~\cite{chen2019novel} (of width $T_v$)
to illustrate how the reconstructed execution intervals are processed.
On a schedule ladder diagram, the victim task's arrivals are always present in the same column (since the width equals $T_v$.) 
Let's define such a column as the ``true arrival column'' that has an offset of $\phi_v$ from the leftmost time column. 
As the reconstructed execution intervals of $\tau_o$ are ensured to have priorities lower than the victim task, those execution intervals will not appear in the true arrival column (because otherwise the victim task would have executed instead). 
In other words, the time columns where the reconstructed execution intervals of $\tau_o$ appear even once cannot be the true arrival column.

Given a reconstructed execution interval $e^r_o \in E^{recon}_o$, the time columns where $e^r_o$ is present are determined by $\{t \mmod T_v \mid begin^r_o \leq t < end^r_o \wedge t \in \mathbb{Z} \}$. 
For simplicity, let's define $e^r_o \mmod T_v := \{t \mmod T_v \mid begin^r_o \leq t < end^r_o \wedge t \in \mathbb{Z}\}$.
Therefore, the time columns where the reconstructed execution intervals appear at least once can be calculated by 
$\bigcup_{e^r_o \in E^{recon}_o} (e^r_o \mmod T_v)$
which represents a set of ``false'' time columns that do not include the true arrival time column.
Thus, the set of candidate time columns can be obtained by
\begin{equation}
\label{eq:compute_candidates}
    \{col \mid 0 \leq col < T_v \wedge col \in \mathbb{Z}\} \ - \hspace{-0.5em}\bigcup_{e^r_o \in E^{recon}_o} \hspace{-0.5em} (e^r_o \mmod T_v)
\end{equation}

\subsection{Inferring The Victim Task's Phase}
\label{sec:infer_phiv}
Next, we take the beginning of the longest contiguous time columns in the candidate list as the inference of the victim task's phase, $\widetilde{\phi_v}$.
Then, the future arrival time of the victim task can be calculated by $\widetilde{\phi_v} + k \cdot T_v$ where $k$ is the desired arrival number.
Alternatively, given a time point $t$, the subsequent arrival time of the victim task can be predicted by 
\begin{equation}
    t + (\widetilde{\phi_v}-t) \mmod T_v
\end{equation}

\begin{figure}[t]
    \centering
    \includegraphics[width=1\columnwidth]{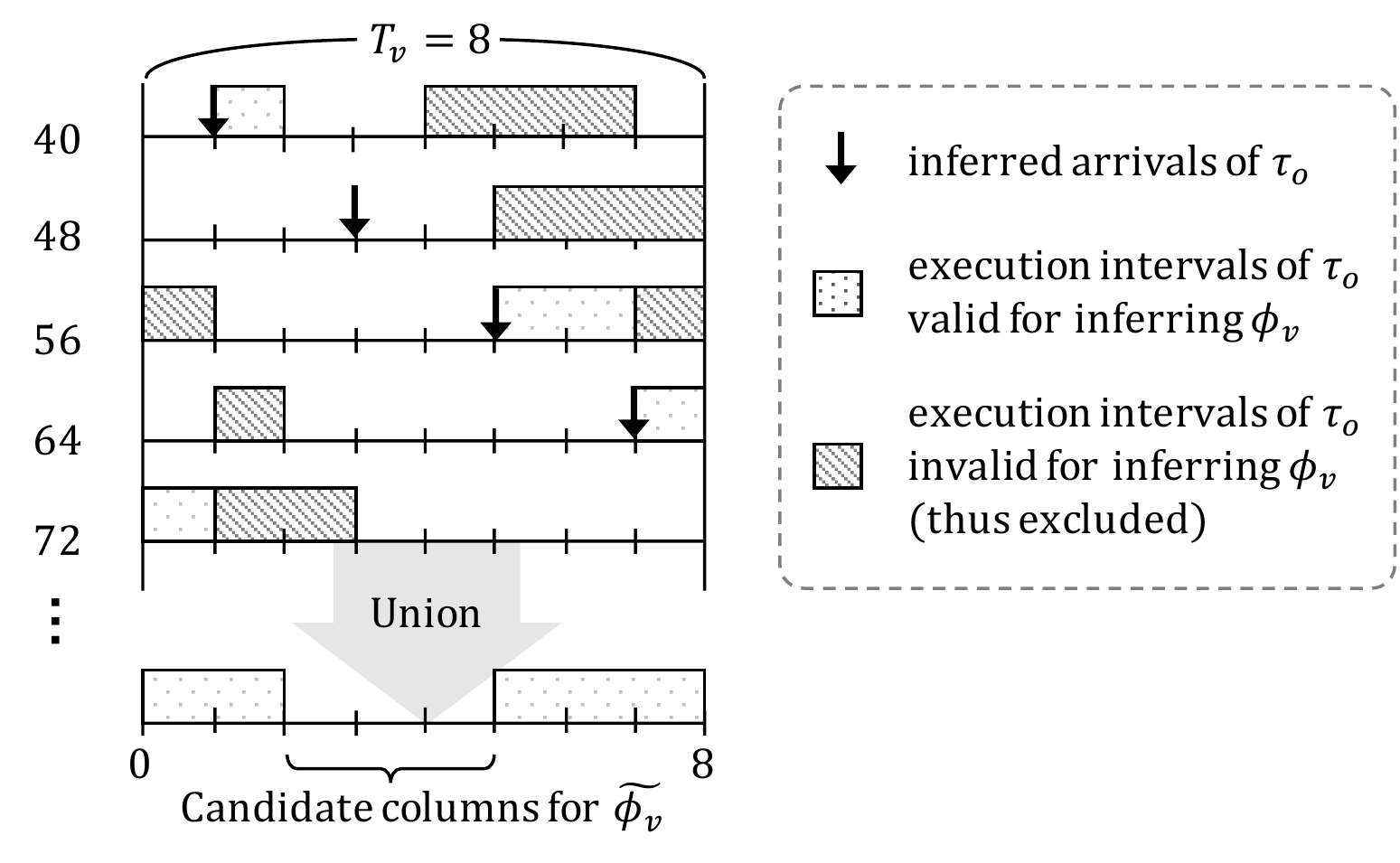}
    \caption{A schedule ladder diagram demonstrated for Example~\ref{ex:infer_phi_v}. Both valid and invalid execution intervals are plotted in the diagram for reference. In the \ATTCKNF algorithms, only the valid execution intervals are taken into account for inferring the victim task's phase (Proposition~\ref{prop:reconstruct_execution_intervals}.)}
    \label{fig:ex_infer_phi_v}
     \vspace{-1.1\baselineskip}
\end{figure}

\begin{example}
\label{ex:infer_phi_v}
Let's consider the task set from Example~\ref{ex:infer_phi_o}.
Assume that the observer task has collected its execution intervals for a duration of $LCM(T_o,T_v)$\footnote{As we will see in Section~\ref{sec:eval_infer_phiv}, the attack duration is evaluated with using $LCM(T_o,T_v)$ as an unit since the offset between the arrivals of $\tau_o$ and $\tau_v$ resets every $LCM(T_o,T_v)$.} (\ie $40$ time units in this example) since the job starts at $t=41$.
The reconstructed execution intervals are $E^{recon}_o=\{[41,42), [61,63), [71,73)\}$ and they correspond to the time columns $\{1\}$, $\{5,6\}$ and $\{0,7\}$, respectively.
Figure~\ref{fig:ex_infer_phi_v} displays the reconstructed execution intervals on a schedule ladder diagram and the timeline at the bottom shows the union of the time columns, $\{0,1,5,6,7\}$, in which the reconstructed execution intervals appear at least once.
The candidate time columns are then computed as $\{0,...,7\}-\{0,1,5,6,7\}=\{2,3,4\}$ (Equation~\ref{eq:compute_candidates}) and the inference is determined as $\widetilde{\phi_v}=2$ (\ie the first time column of the longest contiguous time columns, $\{2,3,4\}$) which matches the ground truth, $\phi_v=2$.
\qedsymbol
\end{example}

It is worth mentioning that the correct $\phi_v$ may not be inferred in Example~\ref{ex:infer_phi_v} if Proposition~\ref{prop:reconstruct_execution_intervals} is not enforced.
If all the execution intervals (both valid and invalid execution intervals in Figure~\ref{fig:ex_infer_phi_v}) are considered when calculating the union of the time columns, the true arrival column will be excluded from the candidate time columns.

The aforementioned situation may happen if an incorrect task phase for the observer task is reconstructed in the first place.
While it may cause some issues when $\widetilde{\phi_o} \neq \phi_o$, our analysis in Section~\ref{sec:analysis_phio} shows that it can happen only under certain rare conditions. The experimental results presented in Section~\ref{sec:eval_infer_phio} further show that the attacker can get a high inference precision even in those conditions due to the presence of run-time variations.


\section{Analysis}
\label{sec:analysis}

\begin{figure}[t]
    \centering
    \includegraphics[width=0.62\columnwidth]{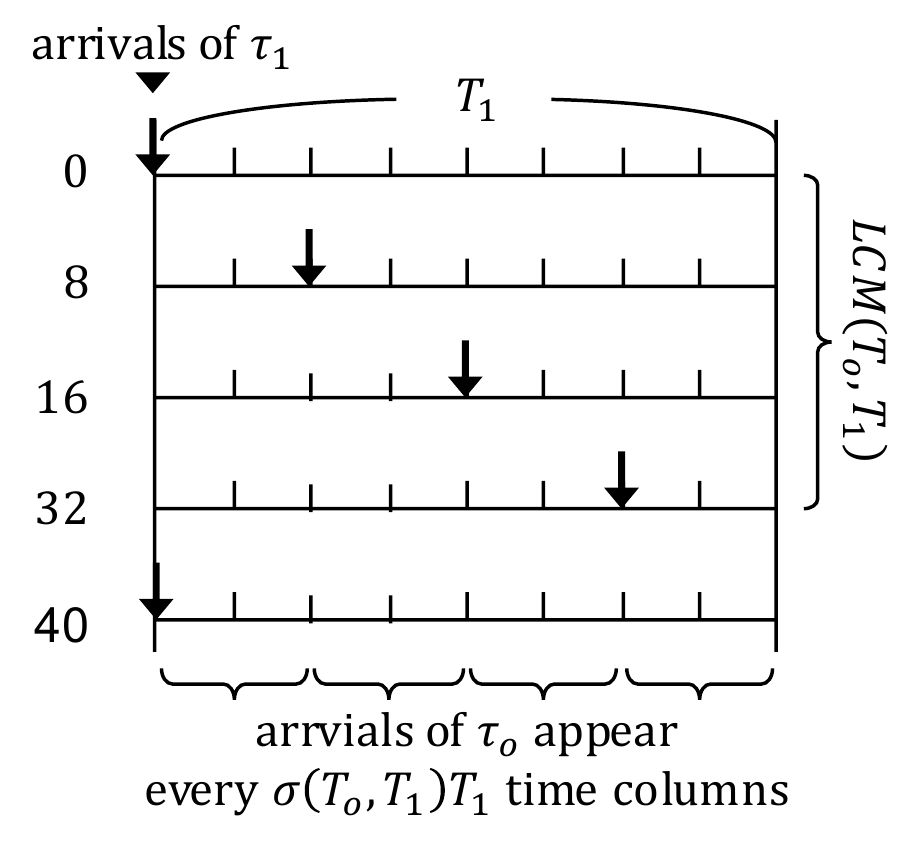}
    \caption{An example of the arrival time correlation between $\tau_o$ and $\tau_1$ for Theorem~\ref{th:a_o_interfered_by_one_task}. In this example, an observer task $\tau_o$ ($T_o=10$) and a periodic task $\tau_1$ ($T_1=8$) are considered. It shows that the projected arrivals of $\tau_o$ appear every $\sigma(T_o,T_1)T_1=2$ time columns.}
    \label{fig:ex_equal_arrival_offset}
     \vspace{-1.1\baselineskip}
\end{figure}

\begin{figure*}[t]
    \centering
    \begin{subfigure}[t]{0.99\textwidth}
        \centering
        \includegraphics[width=0.9\textwidth]{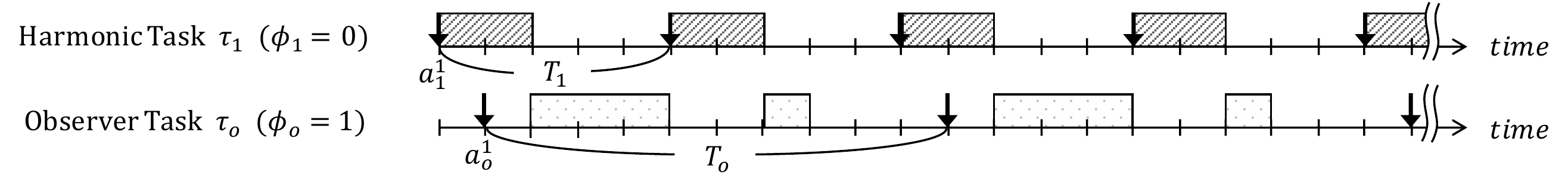}
        \caption{When $\phi_1=0$ and $\phi_o=1$, all the arrivals of $\tau_o$ experience interference due to the harmonic task $\tau_1$.}
    \end{subfigure}%
    \vspace{1\baselineskip}
    \begin{subfigure}[t]{0.99\textwidth}
        \centering
        \includegraphics[width=0.9\textwidth]{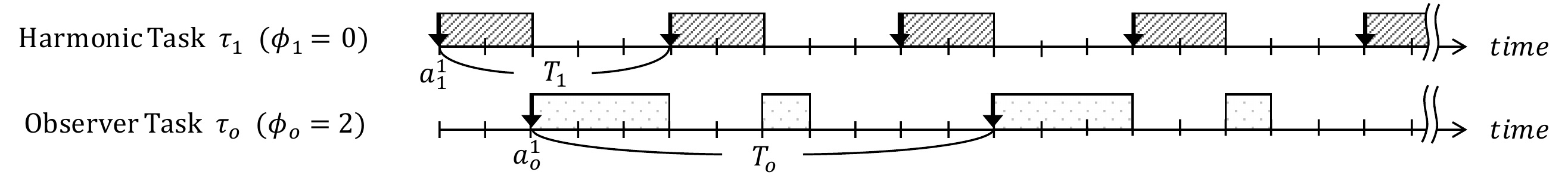}
        \caption{When $\phi_1=0$ and $\phi_o=2$, no arrival of $\tau_o$ experiences interference due to the harmonic task $\tau_1$.}
    \end{subfigure}
    \caption{Schedules of the task set $\Gamma=\{\tau_o, \tau_1\}$ given in Example~\ref{ex:psi_with_harmonic_task}. The $T_o$ and $T_i$ are in harmony and thus $\psi(\tau_o,\tau_1)=1$.}
    \label{fig:ex_interference_of_one_harmonic_task}
     \vspace{-1\baselineskip}
\end{figure*}

\subsection{Impact on The Reconstruction of Observer Task's Phase}
\label{sec:analysis_phio}
In Section~\ref{sec:infer_phi_o} we presented algorithms that can reconstruct the observer task's phase by using the start times of \textit{its own jobs}.
When there exists at least one job of the observer task whose arrival time equals its start time (\ie the job is not delayed by any higher priority task and thus $a_o=s_o$), the correct task phase can be reconstructed.
On the other hand, the correct inference cannot be made if the start of the observer task jobs consistently experience delays in every period. 
Here we explore and characterize the factors that may contribute to the delays in the observer task's start times by analyzing a simple task set with just two tasks.

\begin{theorem}
\label{th:a_o_interfered_by_one_task}
Given an observer task $\tau_o$ and a periodic task $\tau_i$, 
the maximum proportion of the arrivals of $\tau_o$ that may experience interference solely due to task $\tau_i$ is computed by 
\begin{equation}
\label{eq:psi_of_taski}
    \psi(\tau_o,\tau_i)= \left\lceil \frac{u_i}{\sigma_{(\tau_o,\tau_i)}} \right\rceil \sigma_{(\tau_o,\tau_i)}
\end{equation}
where $u_i=\frac{C_i}{T_i}$ is the utilization of $\tau_i$, $\sigma_{(\tau_o,\tau_i)}=\frac{T_o}{LCM(T_o,T_i)}$ is the inverse of the number of the arrivals of $\tau_o$ in a $LCM(T_o,T_i)$ and the resulting $\psi(\tau_o,\tau_i)$ is a fraction in the range  $0 < \psi(\tau_o,\tau_i) \leq 1$.
\end{theorem}
\begin{proof}
Since both tasks are periodic, the schedule of the two tasks repeats after their \textit{least common multiple}, $LCM(T_o, T_i)$. There are $\frac{LCM(T_o,T_i)}{T_o}$ arrivals of $\tau_o$ in a $LCM(T_o, T_i)$. By projecting these arrivals of $\tau_o$ onto a timeline with length equal to the period of $\tau_i$, an arrival of $\tau_o$ appears every $\frac{1}{\frac{LCM(T_o,T_i)}{T_o}} \cdot T_i = \sigma_{(\tau_o,\tau_i)} T_i$ time units and repeats after each $LCM(T_o, T_i)$ \cite[Observation 2]{chen2019novel}, as illustrated in Figure~\ref{fig:ex_equal_arrival_offset}.
The most number of arrivals of $\tau_o$ that may experience interference due to the execution of $\tau_i$ in one $LCM(T_o, T_i)$ is computed by $\left \lceil \frac{C_i}{\sigma_{(\tau_o,\tau_i)} T_i} \right \rceil = \left \lceil \frac{u_i}{\sigma_{(\tau_o,\tau_i)}} \right \rceil$ which happens when there exists an instant at which both tasks arrive at the same time, \eg when $\phi_o = \phi_i$ or $(\phi_o \mmod T_i) = \phi_i$. 
In this case, the maximum proportion of the arrivals of $\tau_o$ may be interfered by $\tau_i$ in one $LCM(T_o, T_i)$ can be computed by 
$\psi(\tau_o,\tau_i) = \left\lceil \frac{u_i}{\sigma_{(\tau_o,\tau_i)}} \right\rceil \sigma_{(\tau_o,\tau_i)}$.
Since this relation is the same across all $LCM(T_o, T_i)$, the calculated $\psi(\tau_o,\tau_i)$ applies to the whole schedule. 
\end{proof}

When $\psi(\tau_o,\tau_i)<1$, it means that at least $1-\psi(\tau_o,\tau_i)$ of the arrivals of $\tau_o$ are not interfered by $\tau_i$. 
In contrast, when $\psi(\tau_o,\tau_i)=1$, it is possible for the execution of $\tau_i$ to interfere all the arrivals of $\tau_o$, which would result in an inaccurate $\widetilde{\phi_o}$.
Intuitively, it can happen when $\tau_i$ has a period in harmony with that of $\tau_o$.  An example is given below.

\begin{example}
\label{ex:psi_with_harmonic_task}
Consider an observer task $\tau_o$ ($T_o=10$, $C_o=4$) and a task $\tau_1$ ($T_1=5$, $C_1=2$). The two tasks are in harmony because $T_o \mmod T_1=0$. Then $\psi(\tau_o,\tau_1)$ is computed by
\begin{equation*}
\sigma_{(\tau_o,\tau_1)}=\frac{T_o}{LCM(T_o,T_1)}=1
\end{equation*}
\begin{equation*}
\psi(\tau_o,\tau_1)= \left\lceil \frac{u_1}{\sigma_{(\tau_o,\tau_1)}} \right\rceil \sigma_{(\tau_o,\tau_1)}=\left\lceil u_1 \right\rceil=1
\end{equation*}
which indicates that all the arrivals of $\tau_o$ may experience interference due to $\tau_1$. 
Figure~\ref{fig:ex_interference_of_one_harmonic_task} illustrates two possible schedules for the given task set. Figure~\ref{fig:ex_interference_of_one_harmonic_task}(a) shows the case where $\phi_o=1$ and $\phi_1=0$, all the arrivals of $\tau_o$ are interfered by $\tau_1$.
Figure~\ref{fig:ex_interference_of_one_harmonic_task}(b) shows the case where $\phi_o=2$ and $\phi_1=0$, all the arrivals of $\tau_o$ are no longer interfered by $\tau_1$.
\qedsymbol

\end{example}

This example shows a crucial fact that $\psi(\tau_o,\tau_i)$ only represents the upper bound of the interference (when only one task is considered). That is, having $\psi(\tau_o,\tau_i)=1$ does not mean all the arrivals of $\tau_o$ are absolutely interfered. 
With the same task set but different task phases (which can vary across systems and every time the system restarts), the impact of the interference can be quite different, as shown in Example~\ref{fig:ex_interference_of_one_harmonic_task}.

On the other hand, the schedules presented in Figure~\ref{fig:ex_interference_of_one_harmonic_task} are generated based on WCETs. However, the actual execution times at run-time in real systems can vary throughout. As a result, the run-time task utilization is in fact smaller and, based on Equation~\ref{eq:psi_of_taski}, the proportion of arrivals of the observer task is also smaller. 
Taking the schedule in Figure~\ref{fig:ex_interference_of_one_harmonic_task}(a) as an example, if the first instance's execution time of $\tau_1$ is $c_1^1=1$ (rather than its WCET), then the arrival $a_o^1$ of the observer task will not experience any interference. In such a case, $s_o^1=a_o^1$ and reconstructing the correct task phase becomes possible.
Therefore, the actual impact is highly dependent on the task phases as well as the run-time variations.

\subsection{Coverage Ratio in EDF}
The coverage ratio $\mathbb{C}(\tau_o,\tau_v)=\frac{C_o}{GCD(T_o,T_v)}$ \cite[Definition 1]{chen2019novel} is used 
to estimate the proportion of the time columns that can be covered by the execution of the observer task in the FP RTS. 
When a given observer task and victim task pair satisfies $\mathbb{C}(\tau_o,\tau_v) \geq 1$, the execution of the observer task may cover all the time columns on the schedule ladder diagram and observing the victim task's arrivals is possible. 
Since the coverage ratio and the inference precision have a positive correlation, it is useful for evaluating the attacker's capabilities against the victim task. 
While the idea of the coverage ratio can be applied in the case of \ATTCKNF, the calculation must be revised to reflect the scope of the reconstructed execution intervals specified in Proposition~\ref{prop:reconstruct_execution_intervals} before it can be used. 
Therefore, we redefine the the coverage ratio for \ATTCKNF as follows:

\begin{definition}
($\mathbb{C}_{\ATTCKNF}(\tau_o,\tau_v)$ \ATTCKNF Coverage Ratio)
The coverage ratio of \ATTCKNF in EDF, denoted by $\mathbb{C}_{\ATTCKNF}(\tau_o,\tau_v)$, is computed by 
\begin{equation}
    \mathbb{C}_{\ATTCKNF}(\tau_o,\tau_v) = \frac{\min(C_o,T_o-T_v)}{GCD(T_o,T_v)}
\end{equation}
%
It represents the proportion of the time columns where the observer task's (valid) execution can
potentially be present in the schedule ladder diagram. If all $T_v$ time columns can be covered by the observer task, then $\mathbb{C}_{\ATTCKNF}(\tau_o,\tau_v) \geq 1$. Otherwise $0 < \mathbb{C}_{\ATTCKNF}(\tau_o,\tau_v) < 1$.
\qedsymbol
\end{definition}

\section{Evaluation}
\label{sec:eval}

\subsection{Evaluation Metrics}
There are mainly two attack stages in the \ATTCKNF algorithms: \ci reconstructing $\phi_o$ for determining valid execution intervals and \cii inferring $\phi_v$ for predicting future arrival time points.
While both stages target the computing of a task's phase, they have very different characteristics due to how they are inferred. We use two different metrics to evaluate the results from the two stages as defined next.

\subsubsection{Reconstructing The Observer Task Phase $\phi_o$}
As introduced in Section~\ref{sec:infer_phi_o}, the observer task's phase is reconstructed based on the collected start times  
where $s_o \geq a_o$ (\ie the start times are always on the right of the corresponding arrival times), thus the distance between $\widetilde{\phi_o}$ and $\phi_o$, denoted by $\Delta \widetilde{\phi_o}=(\widetilde{\phi_o}-\phi_o) \mmod T_o$, should always be positive.
Based on this fact, we define the error ratio for $\widetilde{\phi_o}$ as follows:

\begin{definition}
($\mathbb{E}^o$ Error Ratio of $\widetilde{\phi_o}$)
The error ratio of $\widetilde{\phi_o}$, denoted by $\mathbb{E}^o$, is computed by
\begin{equation}
   \mathbb{E}^o = \frac{\Delta \widetilde{\phi_o}}{T_o}
\end{equation}
where $\Delta \widetilde{\phi_o}=(\widetilde{\phi_o}-\phi_o) \mmod T_o$ represents the distance between $\phi_o$ and a projected $\widetilde{\phi_o}$ on its right. The resulting $\mathbb{E}^o$ value is a real number in the range $0 \leq \mathbb{E}^o \leq 1$.
A smaller $\mathbb{E}^o$ means that the reconstructed $\widetilde{\phi_o}$ has less error when compared to the true $\phi_o$.
\qedsymbol
\end{definition}

Note that $\mathbb{E}^o$ is bounded by $1$ because the start times used for computing $\widetilde{\phi_o}$ are bounded by $a_o \leq s_o < a_o+T_o$.

\subsubsection{Inferring The Victim Task Phase $\phi_v$}
The second stage of the attack in \ATTCKNF is to infer the task phase of the victim task. 
In contrast to the observer task's phase, we are only concerned with how close the inference $\widetilde{\phi_v}$ is to the actual $\phi_v$ and $\widetilde{\phi_v}$ can be on either side of $\phi_v$. 
Inference precision \cite[Definition 2]{chen2019novel} of $\widetilde{\phi_v}$ was introduced to evaluate the effectiveness of scheduler side-channel attacks.
Here, we use the same metric but with a clearer equation for evaluation.
The metric is defined as follows.

\begin{definition}
($\mathbb{I}_v^o$ Inference Precision of $\widetilde{\phi_v}$) 
The inference precision, denoted by $\mathbb{I}_v^o$, is computed by
\begin{equation}
\mathbb{I}_v^o = \left | \frac{ \epsilon }{\frac{T_v}{2}} - 1 \right |
\end{equation}
where $\epsilon = \left | \widetilde{\phi_v} - \phi_v \right |$.
The resulting $\mathbb{I}_v^o$ value is a real number in the range $0 \leq \mathbb{I}_v^o \leq 1$. A larger $\mathbb{I}_v^o$ indicates that the inference $\widetilde{\phi_v}$ is more precise in inferring $\phi_v$.
\qedsymbol
\end{definition}


\subsection{Evaluation Setup}
\label{sec:eval_setup}
The \ATTCKNF algorithms are tested using synthetic task sets that are grouped by utilization from $\{[0.001+0.1\cdot x,0.1+0.1\cdot x) \mid 0 \leq x \leq 9 \wedge x \in \mathbb{Z}\}$. 
Each group contains 6 subgroups that have a fixed number of tasks from $\{5,7,9,11,13,15\}$. A total of 100 task sets are generated for each subgroup.
The utilization of each individual task in a task set is generated from a uniform distribution by using the UUniFast algorithm \cite{uunifast}.
For each task in a task set, the period $T_i$ is randomly drawn from $[100,1000]$ and the worst-case execution time $C_i$ is computed based on the generated task utilization and period.
The task phase is randomly selected from $[0,T_i)$.

The observer task and the victim task in a task set are selected from the generated tasks based on the task periods.
To illustrate, let us consider a task set consisting of $n$ tasks $\Gamma=\{\tau_1,\tau_2,...\tau_n\}$ whose task IDs are ordered by their periods (\ie $T_1>T_2>...>T_n$). The observer task is then selected as the $(\left \lfloor \frac{n}{3} \right \rfloor + 1)$-th task
and the victim task is selected as the $(n- \left \lfloor \frac{n}{3} \right \rfloor)$-th task.
This assignment ensures that $T_o>T_v$ (an assumption from Theorem~\ref{th:ToObserveTv}) and that there exist other tasks with diverse periods (\ie some with smaller periods and some with larger periods compared to $T_o$ and $T_v$.)
The task sets are generated with $\mathbb{C}_{\ATTCKNF}\geq 1$ (except for the experiment that evaluates the impact of the \ATTCKNF coverage ratio). It is to examine the performance of the \ATTCKNF algorithms under the best case (\ie all $T_v$ time columns on the schedule ladder diagram may be covered by the observer task's execution.)

\begin{figure}[t]
    \centering
    \includegraphics[width=0.75\columnwidth]{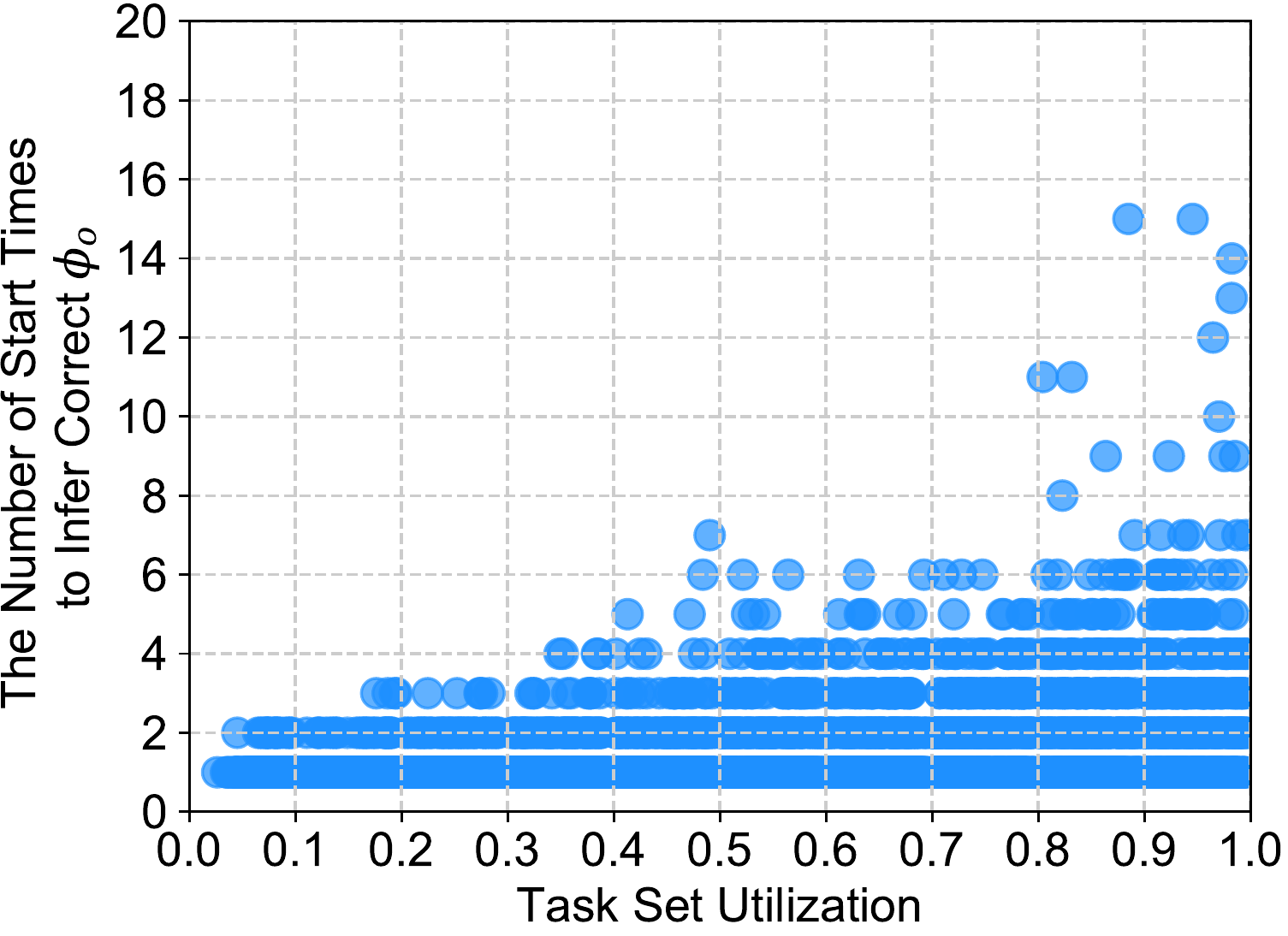}
    \caption{The number of start times the \ATTCKNF algorithms needs for processing each of the 6000 task sets (without harmonic tasks) to get the observer task's correct phase. Each dot in the figure represents the result of a task set. It shows that more start times are needed when the task set utilization is higher. Yet, the worst case (\ie 15 start times which corresponds to $15 \cdot T_o$ attack duration) is reasonably small and manageable.}
    \label{fig:utilvsphio}
    \vspace{-1.2\baselineskip}
\end{figure}

In each task set, $50\%$ of the tasks are configured as sporadic tasks and the rest are periodic tasks. Thus, there are $\left \lfloor \frac{n}{2} \right \rfloor$ sporadic tasks in each task set. The generated task periods are used as the minimum inter-arrival times for the sporadic tasks.
To vary the inter-arrival times for the sporadic tasks at run-time, we use a Poisson distribution in which the probability of each drawn occurrence (\ie each inter-arrival time) is independent.
During the simulation, a Poisson distribution with $\lambda=T_i \cdot 120\%$ as its mean value is used to generate the different inter-arrival times for a sporadic task $\tau_i$.
To satisfy the given minimum inter-arrival time, we adjust the varied inter-arrival time to be $T_i$ if it becomes smaller than $T_i$.

To generate the different task execution times at run-time, a normal distribution is employed.
For a task $\tau_i$, we fit a normal distribution $\mathcal{N}(\mu=C_i\cdot 80\%, \sigma^2)$ with which the cumulative probability for $c_i \leq C_i$ is $99.99\%$.
When a generated execution time exceeds its $C_i$, we adjust it to be $C_i$ to ensure the schedulability of the task set.

\subsection{Simulation Results}

\subsubsection{Reconstructing The Observer Task's Phase $\phi_o$}
\label{sec:eval_infer_phio}

As introduced in Section~\ref{sec:infer_phi_o}, the observer task's phase has to be reconstructed before inferring the victim task's phase. Therefore, we evaluate the factors that impact the reconstruction of the observer task's phase.

We first test the \ATTCKNF algorithms without any tasks that are in harmony with the observer task. 
The result shows that the correct $\phi_o$ (\ie $\widetilde{\phi_o} = \phi_o$ or $\mathbb{E}^o=0$) can be reconstructed in all of the tested 6000 task sets (without harmonic tasks). 
The number of start times that the \ATTCKNF algorithms process to get $\mathbb{E}^o=0$ is plotted in Figure~\ref{fig:utilvsphio}.
The figure shows a trend that it requires more start times for the \ATTCKNF algorithms to reconstruct correct $\phi_o$ when the utilization is higher.
This is because a higher utilization implies higher chances of preemptions and delays in task executions.
Nevertheless, even in the worst case, of 6000 tested task sets, a reasonably small number of start times were required (\ie 15 start times). This shows that the \ATTCKNF algorithms are practical for reconstructing $\phi_o$.

\begin{figure}[t]
    \centering
    \includegraphics[width=0.8\columnwidth]{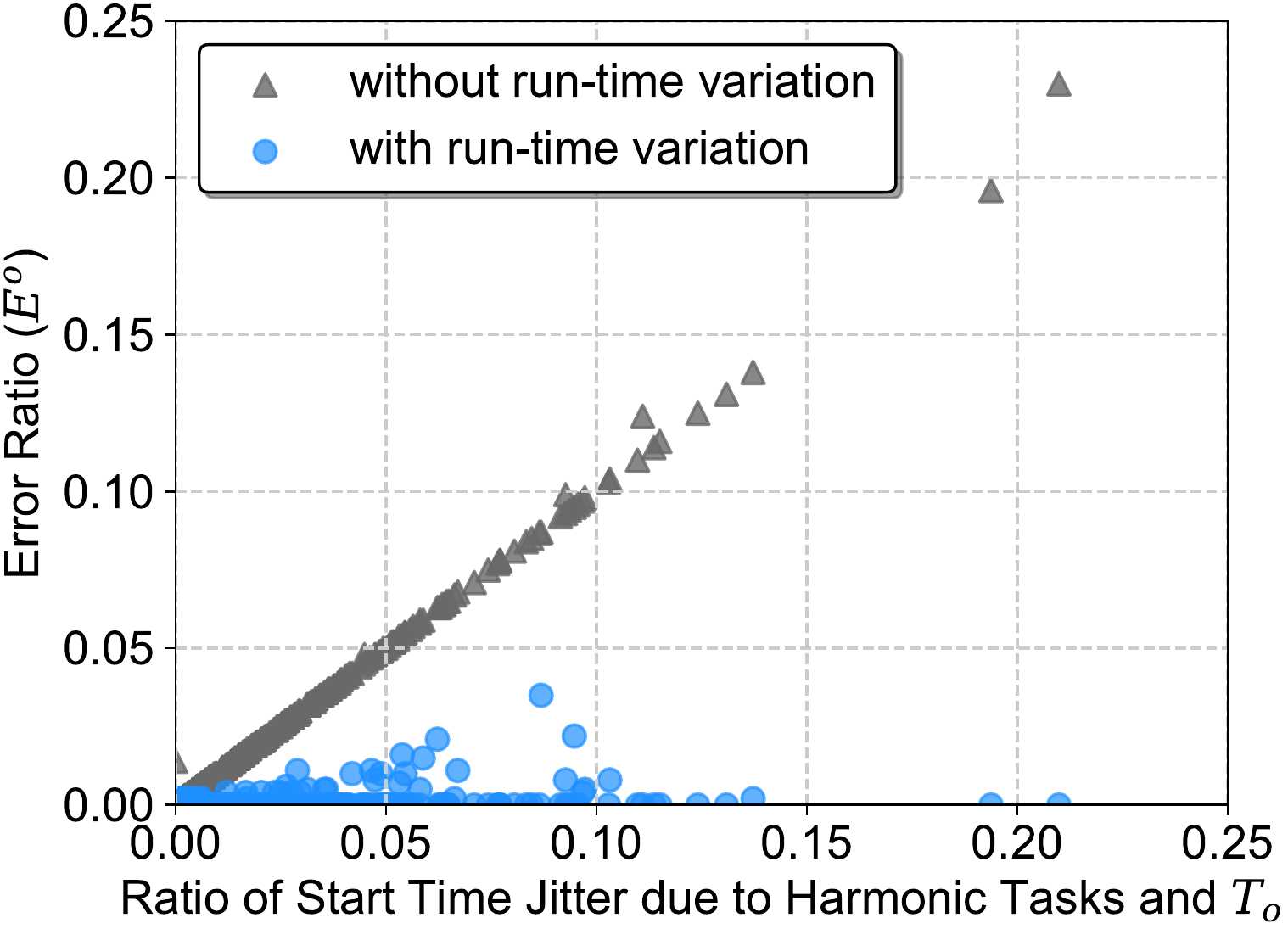}
    \caption{The impact of the start time delays caused by the tasks in harmony with the observer task on error ratio. The tested 6000 task sets are generated with at least one task in harmony with the observer task in each task set. The X-axis is the ratio of the estimated start time delays introduced by the harmonic tasks and the observer task's period and the Y-axis is the error ratio. 
    Each task set is tested with run-time variations and without run-time variations. The results show that the impact of the start time delays becomes subtle in the presence of run-time variations.}
    \label{fig:jittervsphio}
    \vspace{-1.3\baselineskip}
\end{figure}

Next we evaluate the impact of harmonic tasks. As pointed out in Section~\ref{sec:analysis_phio}, a task with a period that is in harmony with the period of the observer task can contribute a constant delay to the observer task's start times. Therefore, we regenerate 6000 task sets with at least one task in harmony with the observer task in each task set and test with the \ATTCKNF algorithms.
In this experiment, we let the \ATTCKNF algorithms collect start times within a duration of $10 \cdot LCM(T_o,T_v)$ (see Section~\ref{sec:eval_infer_phiv} for the explanation for the chosen duration).
Each task set is tested with two conditions: \ci without run-time variations and \cii with run-time variations. 
It is to examine the impact of the constant delay caused by the harmonic tasks in both a static and a more realistic RTS environment.
Without run-time variations, the task's execution times will run up to the WCET in every job instance which should retain the constant delay contributed by the harmonic tasks (or any potential delay contributed by other tasks). 
In contrast, as introduced in Section~\ref{sec:eval_setup}, with run-time variations, the task's execution times include variations that are drawn from a normal distribution that gives a more realistic run-time result.
The results are plotted in Figure~\ref{fig:jittervsphio} and show that, without run-time variations, the error ratio $\mathbb{E}^o$ is proportional to the ratio of the constant delay caused by the harmonic tasks and the observer task's period $T_o$. The outliers for the triangular points are the cases when there are other tasks contributing to the start time delays within the tested attack duration (\ie $10 \cdot LCM(T_o,T_v)$).
It is worth noting that only $6.75\%$ of the task sets have constant delay and $\mathbb{E}^o>0$ since the task phases are randomly generated and hence it is not guaranteed that the harmonic tasks can always interfere with the observer task's arrivals.
This also indicates that having harmonic tasks does not necessarily downgrade the proposed attack.
On the other hand, with execution time variations, the impact of the constant start time delay is significantly reduced due to the varied execution times. The \ATTCKNF algorithms yield better error ratios in all of the 6000 task sets 
($88.4\%$ of the task sets that have $\mathbb{E}^o>0$ without run-time variations yield $\mathbb{E}^o=0$ with run-time variations.)

\begin{figure}[t]
    \centering
    \includegraphics[width=0.75\columnwidth]{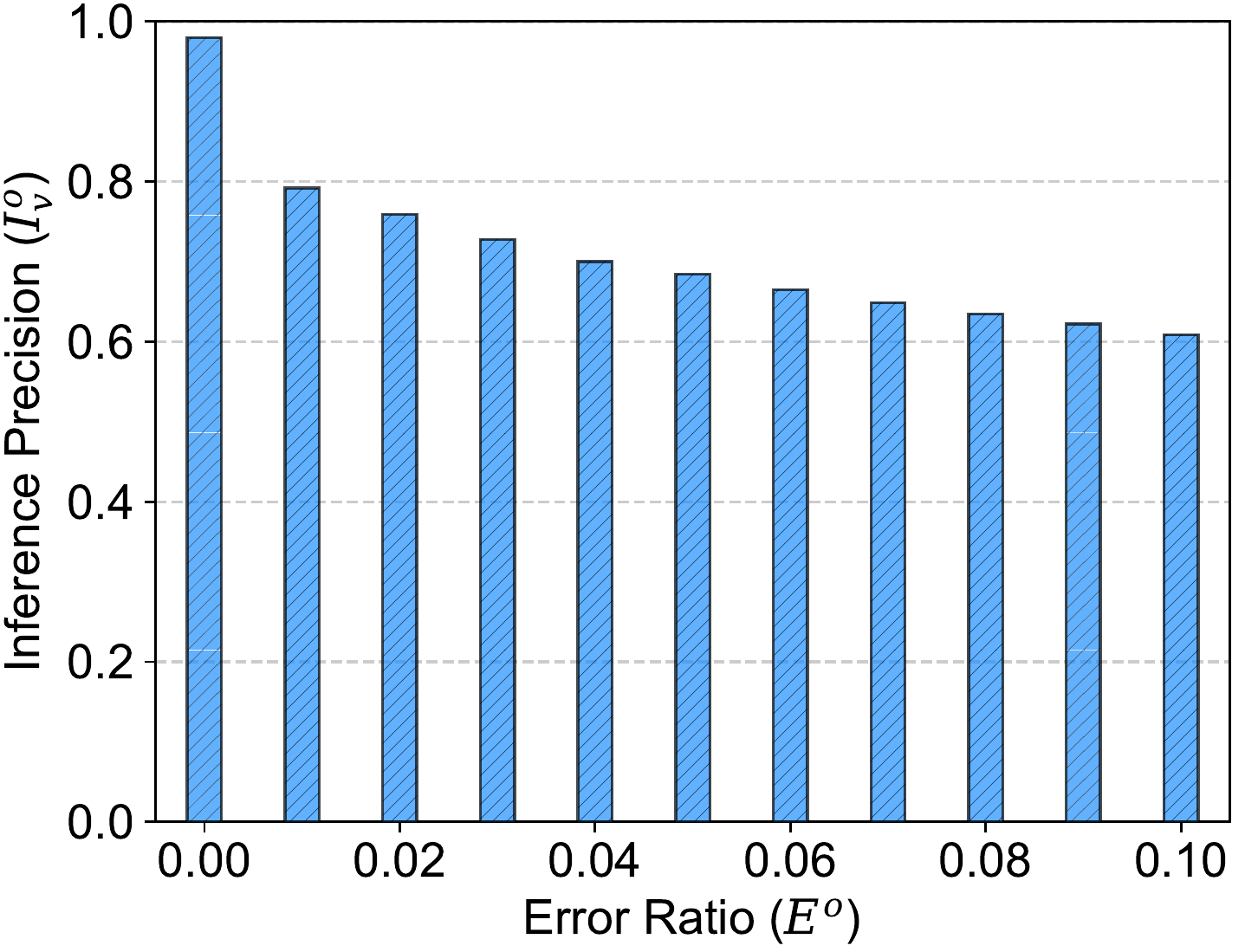}
    \caption{The impact of the error ratio on the inference precision. Each error ratio bin is the result of the 6000 task sets (with harmonic tasks) with synthetically generated $\widetilde{\phi_o}$ from $\{\phi_o + (0.01x)T_o \mid 0 \leq x \leq 10 \wedge x\in \mathbb{Z}\}$. It shows that the error ratio has negative impact on the inference precision.}
    \label{fig:error_vs_precision}
    \vspace{-1.3\baselineskip}
\end{figure}

To understand the impact of the error ratio of $\widetilde{\phi_o}$ on the inference precision of $\widetilde{\phi_v}$, we test the aforementioned 6000 task sets with synthetically generated $\widetilde{\phi_o}$ from $\{\phi_o + (0.01x)T_o \mid 0 \leq x \leq 10 \wedge x\in \mathbb{Z}\}$ which is expected to yield error ratio in $\{0, 0.01,...,0.1\}$. This range is chosen based on the experiment results shown in Figure~\ref{fig:jittervsphio} where overall error ratio is smaller than $0.034$ with run-time variations.
The experiment is carried out with a duration of $10 \cdot LCM(T_o,T_v)$ and with run-time variations enabled. Results shown in Figure~\ref{fig:error_vs_precision} indicate that the error ratio has considerable negative impact on the inference precision.
For example, inference precision $\mathbb{I}^0_v$ reduces to $0.79$ when $\mathbb{E}^o=0.01$ from $\mathbb{I}^0_v=0.98$ when $\mathbb{E}^o=0$. Nevertheless, considering only $0.7\%$ of the task sets have $\mathbb{E}^o>0$ due to the run-time variations and the varied task phases, a high $\mathbb{E}^o$ is arguably uncommon in real cases.

\subsubsection{Inferring The Victim Task's Phase $\phi_v$}
\label{sec:eval_infer_phiv}

\begin{figure}[t]
    \centering
    \includegraphics[width=0.75\columnwidth]{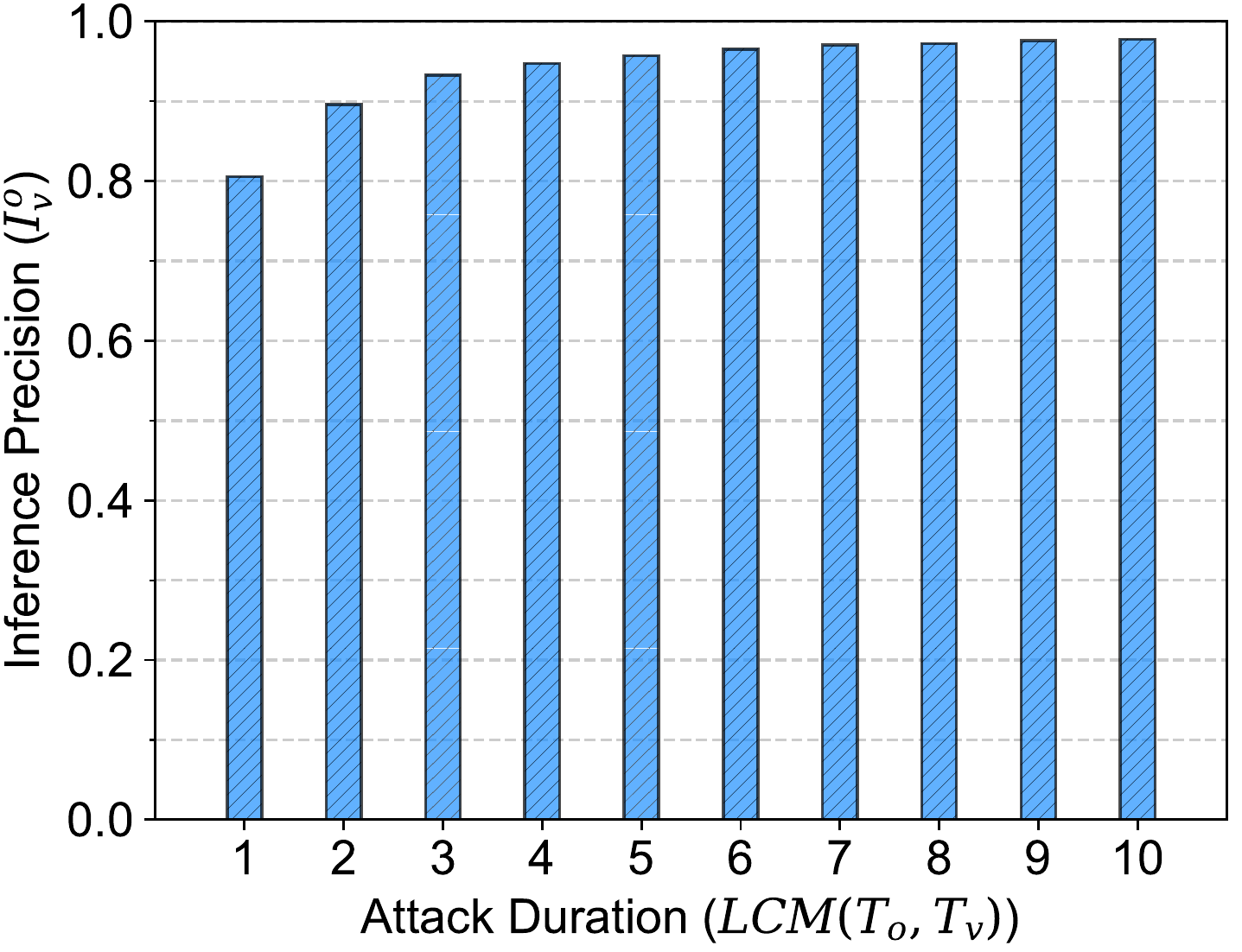}
    \caption{The inference precision with varying attack duration. The X-axis is the length of the attacks based on $LCM(T_o,T_v)$. It shows that the longer the attack persists the higher inference precision \ATTCKNF can achieve.}
    \label{fig:duration}
     \vspace{-0.5\baselineskip}
\end{figure}

We now focus on evaluating the inference precision of $\widetilde{\phi_v}$.
Note that the complete \ATTCKNF algorithms are tested in the experiments here. That is, the $\widetilde{\phi_o}$ values are reconstructed as a part of the algorithms during the experiments.
We first examine the impact of the attack duration on the attack results.
We use $LCM(T_o,T_v)$ as an unit of the attack duration to evaluate \ATTCKNF since the offset between the observer task and the victim task repeats every $LCM(T_o,T_v)$.
A total of 6000 task sets are tested with the attack duration varying from $LCM(T_o,T_v)$ to $10 \cdot LCM(T_o,T_v)$ and the results are plotted in Figure~\ref{fig:duration}.
As shown, a longer attack duration leads to a higher inference precision.
It is because more execution intervals are reconstructed as the attack lasts longer.
The inference precision reaches $\mathbb{I}_v^o=0.978$ at $10\cdot LCM(T_o,T_v)$ and plateaus afterward.
Therefore, we choose to use an attack duration of $10 \cdot LCM(T_o,T_v)$ for the experiments presented in this section unless otherwise stated.


\begin{figure}[t]
    \centering
    \includegraphics[width=0.75\columnwidth]{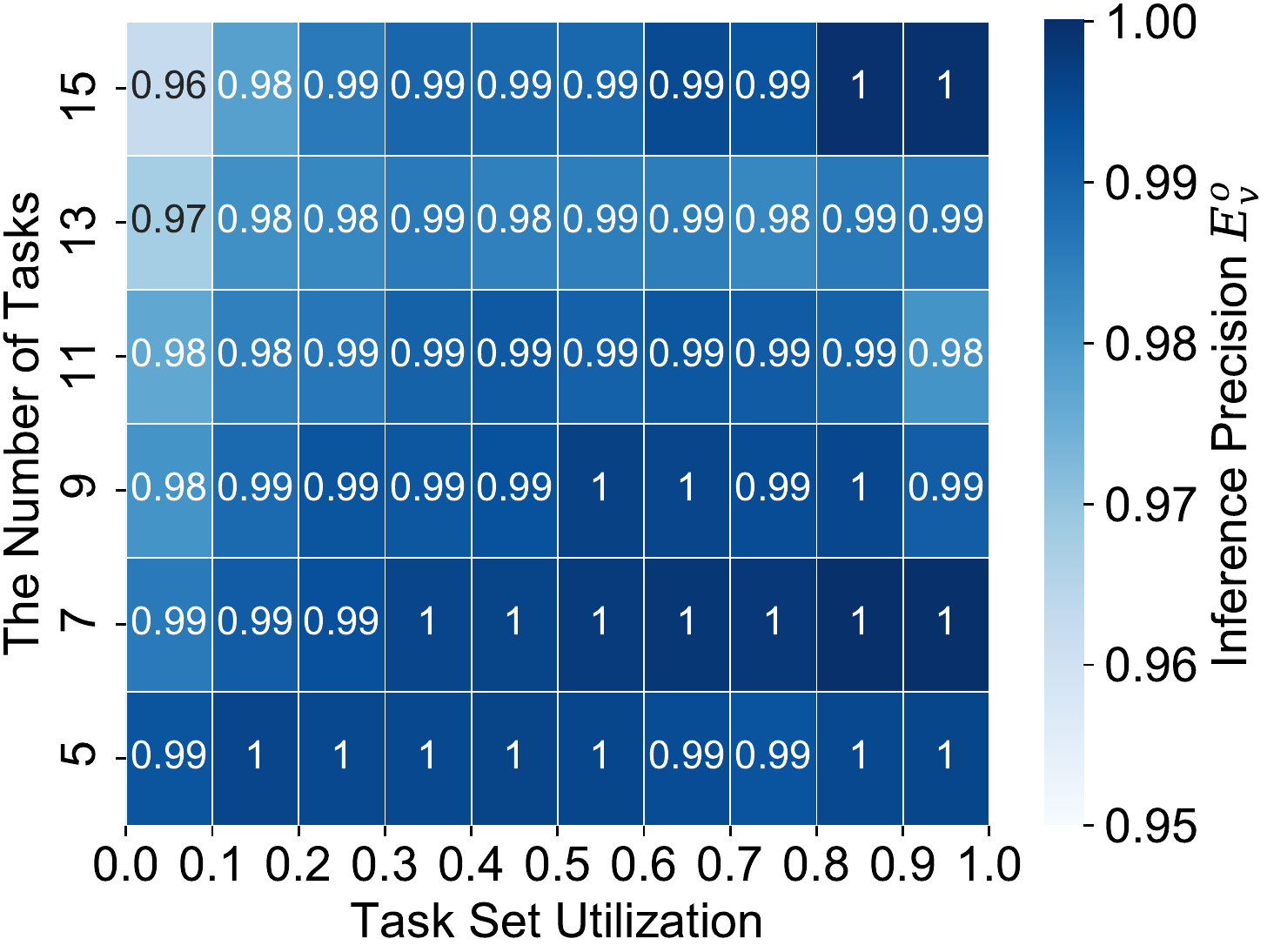}
    \caption{The impact of the number of tasks in a task set and the task set utilization on the inference precision. Each grid displays the mean inference precision for the corresponding number of tasks (Y-axis) and task utilization (X-axis). A darker grid shows a higher inference precision.}
    \label{fig:task_num_vs_util_vs_precision}
     \vspace{-1.2\baselineskip}
\end{figure}

Next we break down the number of tasks in a task set and the task utilization to evaluate their impact on the inference precision.
The experiment results are plotted in Figure~\ref{fig:task_num_vs_util_vs_precision}.
Each grid in the figure shows the mean inference precision of 100 task sets with the corresponding number of tasks in a task set and the task utilization. 
A brighter grid has a lower inference precision while a darker grid indicates that the attack yields a better inference precision.
The resulting heat map gives an intuition for the distribution of the inference precision with varying number of tasks and the task utilization.
The figure shows a small degradation when the number of tasks in a task set is high and the utilization is low.
It is because a task set with a higher number of tasks can have more tasks preempting and delaying the observer task's execution and this introduces more perturbations to the algorithms.  
On the other hand, a low utilization implies a low execution time for the observer task which results in shorter execution intervals to be reconstructed.
This fact makes it difficult to effectively eliminate false time columns and leads to more scattered candidate time slots for the last step of the algorithms, which causes more uncertainty to the inferences.

\subsubsection{Impact of \ATTCKNF Coverage Ratio}

\begin{figure}[t]
    \centering
    \includegraphics[width=0.8\columnwidth]{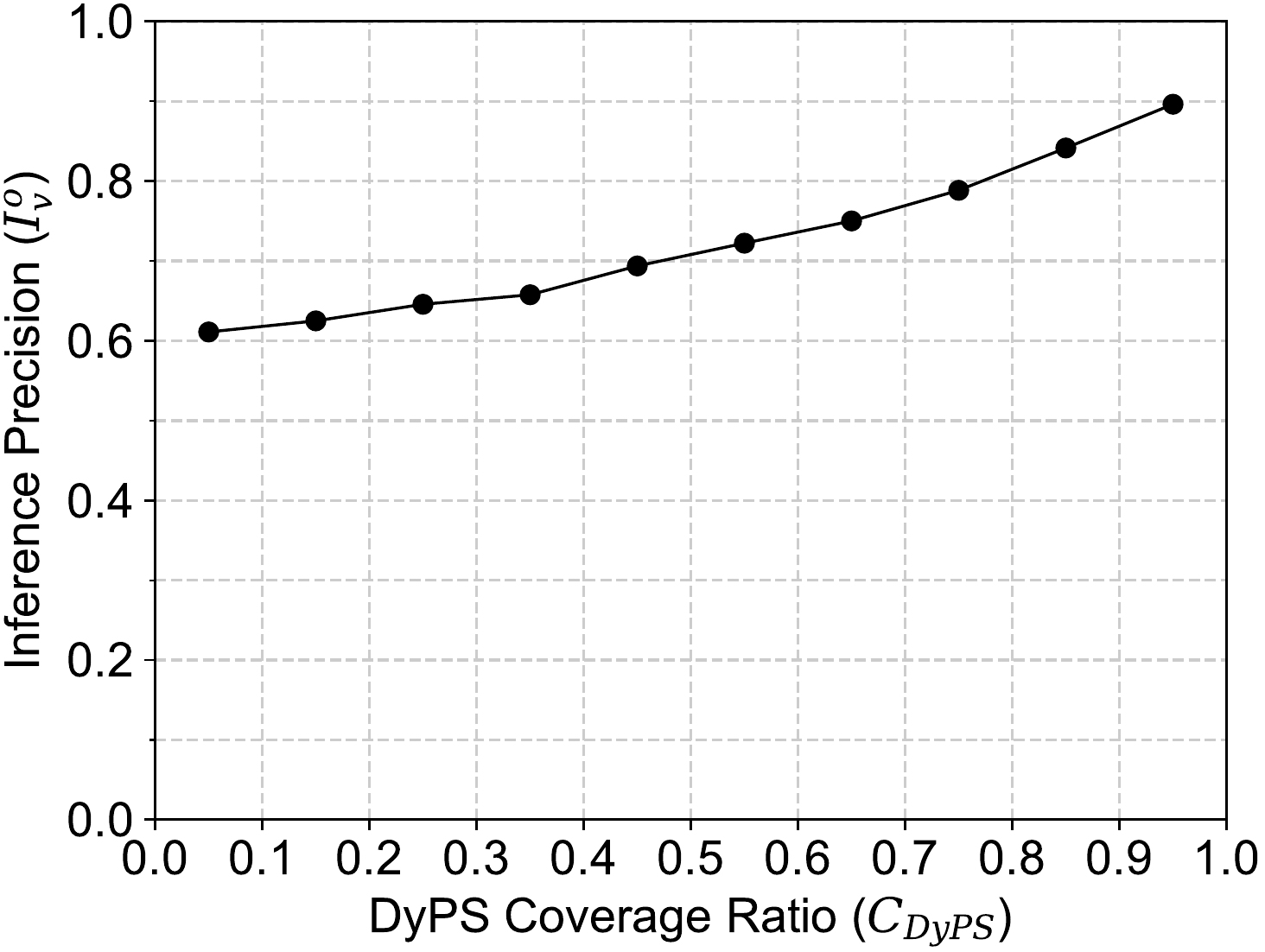}
    \caption{The inference precision with varying the \ATTCKNF coverage ratio in the range of $0 < \mathbb{C}_{\ATTCKNF} < 1$. It suggests that the \ATTCKNF algorithms get better performance as the \ATTCKNF coverage ratio increases. The \ATTCKNF algorithms perform better than a random guess when $\mathbb{C}_{\ATTCKNF} \approx 0$.}
    \label{fig:coverage_ratio}
     \vspace{-1.25\baselineskip}
\end{figure}

The \ATTCKNF coverage ratio $\mathbb{C}_{\ATTCKNF}$ represents the proportion of the time columns in a schedule ladder diagram that can be covered by the observer task's execution.
In the previous experiments with $\mathbb{C}_{\ATTCKNF} \geq 1$, we showed that the \ATTCKNF algorithms are able to yield competitive inference precisions in various task set conditions. 
Here, we evaluate the case when the \ATTCKNF coverage ratio is less than one (\ie $0 < \mathbb{C}_{\ATTCKNF} < 1$.)
In this experiment, we generate 6000 task sets for each of the \ATTCKNF coverage ratio groups from $\{[0.001+0.1\cdot x,0.1+0.1\cdot x) \mid 0 \leq x \leq 9 \wedge x \in \mathbb{Z}\}$. The mean inference precision is then taken from each \ATTCKNF coverage ratio group and the results are plotted in Figure~\ref{fig:coverage_ratio}.
It shows that the \ATTCKNF algorithms get worse results as the \ATTCKNF coverage ratio drops.
When $\mathbb{C}_{\ATTCKNF} \approx 0$, the \ATTCKNF algorithms yield a inference precision ($\mathbb{I}^o_v=0.61$) that is close to, but still better than, a random guess.
Conversely, a higher \ATTCKNF coverage ratio has better inference precision since the observer task can cover more time columns on the schedule ladder diagram and hence has a higher chance to encapsulate the true time column.

\subsubsection{Comparison with State-of-the-art (ScheduLeak)}
The main challenge in attacking the EDF RTS is the existence of the invalid part of the execution intervals.
To understand the impact on the state-of-the-art (\ie ScheduLeak) algorithms and how the \ATTCKNF algorithms handle such intervals, we test both algorithms in the EDF RTS with 6000 newly generated task sets (with $C_o > (T_o-T_v)$, task utilization drawn from $[0.001,1.0)$) in which the presence of the invalid execution intervals is guaranteed. 
Results presented in Figure~\ref{fig:duration_sleak_in_edf} show that the \ATTCKNF algorithms outperform the ScheduLeak algorithms in handling the invalid execution intervals. As the attack lasts longer, more invalid execution intervals are collected in the case of ScheduLeak that leads to worse inference precision. The ScheduLeak algorithms yield a mean inference prevision of $0.54$ at $10 \cdot LCM(T_o,T_v)$. This is only sightly better than a naive attack with random guesses indicating 
ScheduLeak is inadequate for attacking the EDF RTS.


\section{Discussion}

From the evaluation results, \ATTCKNF attacks perform well in inferring the victim task's phases and have similar performance characteristics to ScheduLeak~\cite{chen2019novel}  under various task set conditions.
As a result, given the same task set, running either FP scheduling or EDF scheduling does not seemn to make a significant difference in terms of resisting the attacks. 
However, the \ATTCKNF algorithms do require an extra step to reconstruct the observer task's phase before proceeding to inferring the victim task's phase, which is the crucial point that distinguishes the scheduler side-channels in FP RTS and EDF RTS.
As analyzed in Section~\ref{sec:analysis_phio}, the observer task may suffer constant start time delays when there exist some harmonic tasks that have $\psi(\tau_o,\tau_i)=1$ and $\phi_i$ aligned with $\phi_o$.
Having such tasks may lead to an inaccurate $\widetilde{\phi_o}$ and further reduce the inference precision as evaluated in Section~\ref{sec:eval_infer_phio}.
This offers one potential effective measure for defending EDF RTS against the \ATTCKNF attack. Consequently, employing the EDF scheduling alogrihtm and adjusting the task parameters to satisfy the aforementioned conditions can be a simple yet cost effective defense compared to the schedule obfuscation techniques (proposed for, \eg the FP scheduling~\cite{2016:taskshuffler} and the time-triggered scheduling~\cite{kruger2018vulnerability}) that require fundamental changes in the design of the schedulers.
\section{Related Work}

Side-channels and covert-channels have been well studied in many research domains.
Typical side-channels such as cache access time~\cite{kelsey1998side}, power consumption traces~\cite{Jiang2014}, electromagnetic emanations~\cite{agrawal2002side}, temperature~\cite{bar2006sorcerer}, \etc, may also be seen against RTS.
In this paper, we focus on scheduler-based channels.
V\"{o}lp \emph{et al.} \cite{embeddedsecurity:volp2008} examined covert channels between tasks with different priorities in the FP RTS. The authors proposed to modify the scheduler to alter the task that could potentially leak information with the idle task to avoid such covert channels. 
Kadloor \emph{et al.} \cite{kadloor2013} presented a methodology for quantifying side-channel leakage for first-come-first-serve and time-division-multiple-access schedulers. Gong and Kiyavash \cite{GongK14} analyzed the deterministic work-conserving schedulers and discovered a lower bound for the total information leakage. 
There also has been other work on covert channels and side-channels in RTS~\cite{tsalis2019taxonomy, ghassami2015capacity, embeddedsecurity:son2006,embeddedsecurity:volp2013}.
In contrast to the above side-channels and the existing covert-channels, this paper focuses on the scheduler side-channels in which an unprivileged task is able to learn the behaviors of others by simply analyzing its own execution.

There exists a line of work that focus on the scheduler side-channels in RTS.
Chen \etal~\cite{chen2019novel} first introduced the scheduler side-channels in preemptive FP RTS and proposed the ScheduLeak algorithms that can extract valuable task information from the system schedules at run-time with using an unprivileged user-space task. 
Liu \etal~\cite{liu2019leaking} proposed a method to obtain the task parameters (\ie the victim task's period) that are required in the scheduler side-channel attacks.
Yoon \etal~\cite{2016:taskshuffler} attempted to close the scheduler side-channels by introducing a randomization protocol that obfuscates the schedules in the FP RTS.  
Based on this idea, Kr{\"u}ger \etal~\cite{kruger2018vulnerability} proposed a combined online/offline randomization scheme to reduce determinisms for time-triggered systems. 
Nasri \etal~\cite{nasri2019pitfalls} conducted a comprehensive study on the schedule randomization protocol and argued that such techniques can expose the FP RTS to more risks.
While these work are centered in the problem of the scheduler side-channels, none of them focused on the dynamic-priority RTS and thus the above defense techniques are inapplicable to preventing the proposed \ATTCKNF attack. 

There has been some work on security integration in RTS~\cite{easwaran2017systematic,chen2018securing,YoonMHM:2015,xie2007schedulesecurity,lin2009rtssecurity,trilla2018cache}. 
Most of the work focused on defence techniques against general attacks. 
Some researchers framed security in RTS as a scheduling problem~\cite{embeddedsecurity:mohan2014,embeddedsecurity:mohan2015}. 
Hasan \etal~\cite{hasan2016exploring} discussed the considerations of scheduling security tasks in legacy RTS.
This is further extended to integrating security tasks into more general RTS models~\cite{hasan2018design,hasan2017contego}.
There exist another group of works focused on hardening RTS from the architecture perspective.
Mohan \etal~\cite{mohan2013s3a} proposed to use a disjoint, trusted hardware component (\ie FPGA) to monitor the behavior of a real-time program running on an untrustworthy RTS.
Yoon \etal~\cite{yoon2013securecore} created the SecureCore framework that utilizes one of the cores in a multi-core processor as a trusted entity to carry out various security checks for the activities observed from other cores.
Abdi \etal~\cite{abdi2018guaranteed,abdi2018preserving} developed a restart-based approach that uses a root-of-trust (\ie a piece of hardware circuit) and the trust zone technology to enforce the system reboot process to evict any malicious dwellers when necessary.
While some of the techniques are useful for detecting anomalies and mitigating the impact of the attacks, they do not close the scheduler side-channels presented in this paper.

\section{Conclusion}
\label{sec::concl}

Dynamic priority scheduling algorithms such as EDF are able to better optimize the resources in the system (as compared to FP schedulers). In addition, they also increase the ``dynamicism'' in the system that can increase the difficulty for would-be-attackers looking to leak critical information via side-channels. Though not immune from such attacks, employing harmonic tasks with certain characteristics seems to have the potential to increase the difficulty for attackers. Hence, an increased adoption of dynamic scheduler, along with the aforementioned design choices, in RTS may aid in improving the security posture of such systems. 

\newpage
\bibliographystyle{IEEEtran}

\bibliography{ref}

\begin{thebibliography}{10}
\providecommand{\url}[1]{#1}
\csname url@samestyle\endcsname
\providecommand{\newblock}{\relax}
\providecommand{\bibinfo}[2]{#2}
\providecommand{\BIBentrySTDinterwordspacing}{\spaceskip=0pt\relax}
\providecommand{\BIBentryALTinterwordstretchfactor}{4}
\providecommand{\BIBentryALTinterwordspacing}{\spaceskip=\fontdimen2\font plus
\BIBentryALTinterwordstretchfactor\fontdimen3\font minus
  \fontdimen4\font\relax}
\providecommand{\BIBforeignlanguage}[2]{{%
\expandafter\ifx\csname l@#1\endcsname\relax
\typeout{** WARNING: IEEEtran.bst: No hyphenation pattern has been}%
\typeout{** loaded for the language `#1'. Using the pattern for}%
\typeout{** the default language instead.}%
\else
\language=\csname l@#1\endcsname
\fi
#2}}
\providecommand{\BIBdecl}{\relax}
\BIBdecl

\bibitem{chen2011stuxnet}
T.~M. Chen and S.~Abu-Nimeh, ``Lessons from stuxnet,'' \emph{Computer},
  vol.~44, no.~4, pp. 91--93, Apr. 2011.

\bibitem{case2016analysis}
D.~U. Case, ``Analysis of the cyber attack on the ukrainian power grid,''
  \emph{Electricity Information Sharing and Analysis Center (E-ISAC)}, 2016.

\bibitem{JeepHacking101}
``{Jeep Hacking 101},'' \emph{IEEE Spectrum}, Aug 2015,
  \url{http://spectrum.ieee.org/cars-that-think/transportation/systems/jeep-hacking-101}.

\bibitem{byungho2014attack}
B.~Min and V.~Varadharajan, ``Design and analysis of security attacks against
  critical smart grid infrastructures,'' \emph{2014 19th International
  Conference on Engineering of Complex Computer Systems}, vol.~0, pp. 59--68,
  2014.

\bibitem{yoon2017virtualdrone}
M.-K. Yoon, B.~Liu, N.~Hovakimyan, and L.~Sha, ``Virtualdrone: virtual sensing,
  actuation, and communication for attack-resilient unmanned aerial systems,''
  in \emph{Proceedings of the 8th International Conference on Cyber-Physical
  Systems}.\hskip 1em plus 0.5em minus 0.4em\relax ACM, 2017, pp. 143--154.

\bibitem{auto:koscher2010}
K.~Koscher, A.~Czeskis, F.~Roesner, S.~Patel, T.~Kohno, S.~Checkoway, D.~McCoy,
  B.~Kantor, D.~Anderson, H.~Shacham, and S.~Savage, ``Experimental security
  analysis of a modern automobile,'' in \emph{Security and Privacy (SP), 2010
  IEEE Symposium on}, may 2010, pp. 447 --462.

\bibitem{DroneHack:Shepard2012}
D.~Shepard, J.~Bhatti, and T.~Humphreys, ``Drone hack: Spoofing attack
  demonstration on a civilian unmanned aerial vehicle,'' \emph{GPS World},
  August 2012.

\bibitem{embeddedsecurity:teso2013}
H.~Teso, ``Aicraft hacking,'' in \emph{Fourth Annual {HITB} Security Conference
  in Europe}, 2013.

\bibitem{2016:taskshuffler}
M.~Yoon, S.~Mohan, C.~Chen, and L.~Sha, ``Taskshuffler: A schedule
  randomization protocol for obfuscation against timing inference attacks in
  real-time systems,'' in \emph{2016 IEEE Real-Time and Embedded Technology and
  Applications Symposium (RTAS)}, April 2016, pp. 1--12.

\bibitem{kruger2018vulnerability}
K.~Kr{\"{u}}ger, M.~V{\"{o}}lp, and G.~Fohler, ``Vulnerability analysis and
  mitigation of directed timing inference based attacks on time-triggered
  systems,'' in \emph{30th Euromicro Conference on Real-Time Systems (ECRTS)},
  2018, pp. 22:1--22:17.

\bibitem{nasri2019pitfalls}
M.~Nasri, T.~Chantem, G.~Bloom, and R.~M. Gerdes, ``On the pitfalls and
  vulnerabilities of schedule randomization against schedule-based attacks,''
  in \emph{2019 IEEE Real-Time and Embedded Technology and Applications
  Symposium (RTAS)}.\hskip 1em plus 0.5em minus 0.4em\relax IEEE, 2019, pp.
  103--116.

\bibitem{chen2019novel}
C.-Y. Chen, S.~Mohan, R.~Pellizzoni, R.~B. Bobba, and N.~Kiyavash, ``A novel
  side-channel in real-time schedulers,'' in \emph{2019 IEEE Real-Time and
  Embedded Technology and Applications Symposium (RTAS)}.\hskip 1em plus 0.5em
  minus 0.4em\relax IEEE, 2019, pp. 90--102.

\bibitem{liu2019leaking}
S.~Liu, N.~Guan, D.~Ji, W.~Liu, X.~Liu, and W.~Yi, ``Leaking your engine speed
  by spectrum analysis of real-time scheduling sequences,'' \emph{Journal of
  Systems Architecture}, 2019.

\bibitem{tankard2011advanced}
C.~Tankard, ``Advanced persistent threats and how to monitor and deter them,''
  \emph{Network Security}, vol. 2011, no.~8, pp. 16--19, 2011.

\bibitem{virvilis2013big}
N.~Virvilis and D.~Gritzalis, ``The big four - what we did wrong in advanced
  persistent threat detection?'' in \emph{2013 International Conference on
  Availability, Reliability and Security}, Sep. 2013, pp. 248--254.

\bibitem{isovic2001handling}
D.~Isovic, \emph{Handling Sporadic Tasks in Real-time Systems: Combined Offline
  and Online Approach}.\hskip 1em plus 0.5em minus 0.4em\relax M{\"a}lardalen
  University, 2001.

\bibitem{LiuLayland1973}
C.~L. Liu and J.~W. Layland, ``{Scheduling algorithms for multiprogramming in a
  hard real-time environment},'' \emph{Journal of the ACM}, 1973.

\bibitem{uunifast}
E.~Bini and G.~C. Buttazzo, ``Measuring the performance of schedulability
  tests,'' \emph{RTS Journal}, vol.~30, no. 1-2, pp. 129--154, 2005.

\bibitem{kelsey1998side}
J.~Kelsey, B.~Schneier, D.~Wagner, and C.~Hall, ``Side channel cryptanalysis of
  product ciphers,'' in \emph{European Symposium on Research in Computer
  Security}, 1998.

\bibitem{Jiang2014}
K.~Jiang, L.~Batina, P.~Eles, and Z.~Peng, ``Robustness analysis of real-time
  scheduling against differential power analysis attacks,'' in \emph{2014 IEEE
  Computer Society Annual Symposium on VLSI (ISVLSI)}, July 2014, pp. 450--455.

\bibitem{agrawal2002side}
D.~Agrawal, B.~Archambeault, J.~R. Rao, and P.~Rohatgi, ``The em side—channel
  (s),'' in \emph{International Workshop on Cryptographic Hardware and Embedded
  Systems}.\hskip 1em plus 0.5em minus 0.4em\relax Springer, 2002, pp. 29--45.

\bibitem{bar2006sorcerer}
H.~Bar-El, H.~Choukri, D.~Naccache, M.~Tunstall, and C.~Whelan, ``The
  sorcerer's apprentice guide to fault attacks,'' \emph{Proceedings of the
  IEEE}, vol.~94, no.~2, pp. 370--382, 2006.

\bibitem{embeddedsecurity:volp2008}
M.~V\"{o}lp, C.-J. Hamann, and H.~H\"{a}rtig, ``Avoiding timing channels in
  fixed-priority schedulers,'' in \emph{Proceedings of the 2008 ACM Symposium
  on Information, Computer and Communications Security (ASIACCS)}, 2008, pp.
  44--55.

\bibitem{kadloor2013}
S.~Kadloor, N.~Kiyavash, and P.~Venkitasubramaniam, ``Mitigating timing side
  channel in shared schedulers,'' \emph{IEEE/ACM Transactions on Networking},
  vol.~24, no.~3, pp. 1562--1573, June 2016.

\bibitem{GongK14}
X.~Gong and N.~Kiyavash, ``Quantifying the information leakage in timing side
  channels in deterministic work-conserving schedulers,'' \emph{IEEE/ACM Trans.
  Netw.}, vol.~24, no.~3, pp. 1841--1852, Jun. 2016.

\bibitem{tsalis2019taxonomy}
N.~Tsalis, E.~Vasilellis, D.~Mentzelioti, and T.~Apostolopoulos, ``A taxonomy
  of side channel attacks on critical infrastructures and relevant systems,''
  in \emph{Critical Infrastructure Security and Resilience}.\hskip 1em plus
  0.5em minus 0.4em\relax Springer, 2019, pp. 283--313.

\bibitem{ghassami2015capacity}
A.~Ghassami, X.~Gong, and N.~Kiyavash, ``Capacity limit of queueing timing
  channel in shared fcfs schedulers,'' in \emph{2015 IEEE International
  Symposium on Information Theory (ISIT)}, June 2015, pp. 789--793.

\bibitem{embeddedsecurity:son2006}
J.~Son and Alves-Foss, ``Covert timing channel analysis of rate monotonic
  real-time scheduling algorithm in mls systems,'' in \emph{2006 IEEE
  Information Assurance Workshop}, June 2006, pp. 361--368.

\bibitem{embeddedsecurity:volp2013}
M.~Völp, B.~Engel, C.~Hamann, and H.~Härtig, ``On confidentiality-preserving
  real-time locking protocols,'' in \emph{2013 IEEE 19th Real-Time and Embedded
  Technology and Applications Symposium (RTAS)}, April 2013, pp. 153--162.

\bibitem{easwaran2017systematic}
A.~Easwaran, A.~Chattopadhyay, and S.~Bhasin, ``A systematic security analysis
  of real-time cyber-physical systems,'' in \emph{2017 22nd Asia and South
  Pacific Design Automation Conference (ASP-DAC)}.\hskip 1em plus 0.5em minus
  0.4em\relax IEEE, 2017, pp. 206--213.

\bibitem{chen2018securing}
C.-Y. Chen, M.~Hasan, and S.~Mohan, ``Securing real-time internet-of-things,''
  \emph{Sensors}, vol.~18, no.~12, p. 4356, 2018.

\bibitem{YoonMHM:2015}
M.~Yoon, S.~Mohan, J.~Choi, and L.~Sha, ``Memory heat map: Anomaly detection in
  real-time embedded systems using memory behavior,'' in \emph{2015 52nd
  ACM/EDAC/IEEE Design Automation Conference (DAC)}, June 2015, pp. 1--6.

\bibitem{xie2007schedulesecurity}
T.~Xie and X.~Qin, ``Improving security for periodic tasks in embedded systems
  through scheduling,'' \emph{ACM Transactions on Embedded Computing Systems
  (TECS)}, vol.~6, no.~3, Jul. 2007.

\bibitem{lin2009rtssecurity}
M.~Lin, L.~Xu, L.~T. Yang, X.~Qin, N.~Zheng, Z.~Wu, and M.~Qiu, ``Static
  security optimization for real-time systems,'' \emph{IEEE Transactions on
  Industrial Informatics}, vol.~5, no.~1, pp. 22--37, Feb 2009.

\bibitem{trilla2018cache}
D.~Trilla, C.~Hernandez, J.~Abella, and F.~J. Cazorla, ``Cache side-channel
  attacks and time-predictability in high-performance critical real-time
  systems,'' in \emph{2018 55th ACM/ESDA/IEEE Design Automation Conference
  (DAC)}, June 2018, pp. 1--6.

\bibitem{embeddedsecurity:mohan2014}
S.~Mohan, M.~K. Yoon, R.~Pellizzoni, and R.~Bobba, ``Real-time systems security
  through scheduler constraints,'' in \emph{2014 26th Euromicro Conference on
  Real-Time Systems (ECRTS)}, July 2014, pp. 129--140.

\bibitem{embeddedsecurity:mohan2015}
R.~Pellizzoni, N.~Paryab, M.~Yoon, S.~Bak, S.~Mohan, and R.~B. Bobba, ``A
  generalized model for preventing information leakage in hard real-time
  systems,'' in \emph{21st IEEE Real-Time and Embedded Technology and
  Applications Symposium (RTAS)}, April 2015, pp. 271--282.

\bibitem{hasan2016exploring}
M.~Hasan, S.~Mohan, R.~B. Bobba, and R.~Pellizzoni, ``Exploring opportunistic
  execution for integrating security into legacy hard real-time systems,'' in
  \emph{2016 IEEE Real-Time Systems Symposium (RTSS)}.\hskip 1em plus 0.5em
  minus 0.4em\relax IEEE, pp. 123--134.

\bibitem{hasan2018design}
M.~Hasan, S.~Mohan, R.~Pellizzoni, and R.~B. Bobba, ``A design-space
  exploration for allocating security tasks in multicore real-time systems,''
  in \emph{2018 Design, Automation \& Test in Europe Conference \& Exhibition
  (DATE)}.\hskip 1em plus 0.5em minus 0.4em\relax IEEE, 2018, pp. 225--230.

\bibitem{hasan2017contego}
------, ``Contego: An adaptive framework for integrating security tasks in
  real-time systems,'' in \emph{29th Euromicro Conference on Real-Time Systems
  (ECRTS 2017)}, 2017.

\bibitem{mohan2013s3a}
S.~Mohan, S.~Bak, E.~Betti, H.~Yun, L.~Sha, and M.~Caccamo, ``S3a: Secure
  system simplex architecture for enhanced security and robustness of
  cyber-physical systems,'' in \emph{Proceedings of the 2nd ACM international
  conference on High confidence networked systems}.\hskip 1em plus 0.5em minus
  0.4em\relax ACM, 2013, pp. 65--74.

\bibitem{yoon2013securecore}
M.-K. Yoon, S.~Mohan, J.~Choi, J.-E. Kim, and L.~Sha, ``Securecore: A
  multicore-based intrusion detection architecture for real-time embedded
  systems,'' in \emph{2013 IEEE 19th Real-Time and Embedded Technology and
  Applications Symposium (RTAS)}.\hskip 1em plus 0.5em minus 0.4em\relax IEEE,
  2013, pp. 21--32.

\bibitem{abdi2018guaranteed}
F.~Abdi, C.-Y. Chen, M.~Hasan, S.~Liu, S.~Mohan, and M.~Caccamo, ``Guaranteed
  physical security with restart-based design for cyber-physical systems,'' in
  \emph{Proceedings of the 9th ACM/IEEE International Conference on
  Cyber-Physical Systems}.\hskip 1em plus 0.5em minus 0.4em\relax IEEE Press,
  2018, pp. 10--21.

\bibitem{abdi2018preserving}
------, ``Preserving physical safety under cyber attacks,'' \emph{IEEE Internet
  of Things Journal}, 2018.

\end{thebibliography}

\end{document}